\documentclass[11pt]{article}

\usepackage{fullpage}

 \usepackage{graphicx}

\usepackage{amssymb,amsmath,eepic,graphics,color}
\usepackage{epsfig}
\usepackage[ruled,vlined,titlenumbered,algo2e]{algorithm2e}

\newenvironment{proof}{\noindent\textit{Proof: }}{{\hfill $\Box$}}

\newtheorem{observation}{Observation}
\newtheorem{definition}{Definition}
\newtheorem{lemma}{Lemma}
\newtheorem{theorem}{Theorem}
\newtheorem{corollary}{Corollary}
\newtheorem{notation}[theorem]{Notation}

\begin{document}

\title{\textbf{A Survey on Algorithmic Aspects of \\ Modular Decomposition}\thanks{Work supported by the French research grant ANR-06-BLAN-0148-01 ``\textit{Graph Decompositions and Algorithms} - \textsc{graal}".}}

\author{Michel Habib\thanks{LIAFA, Universit\'e Paris 7 Diderot, France} \and Christophe Paul\thanks{CNRS - LIRMM, Universit\'e de Montpellier 2, France} }

\date{\today}

\maketitle

\begin{abstract}
The modular decomposition is a technique that applies but is not restricted to graphs. The notion of module naturally appears in the proofs of many graph theoretical theorems. Computing the modular decomposition tree is an important preprocessing step to solve a large number of combinatorial optimization problems. Since the first polynomial time algorithm in the early 70's, the algorithmic of the modular decomposition has known an important development. This paper survey the ideas and techniques that arose from this line of research.
\end{abstract}

\newcommand{\overlap}{\perp}
\newcommand{\lca}{lca}
\renewcommand{\P}{\mathcal{P}}
\newcommand{\M}{\mathcal{M}}
\newcommand{\X}{\mathcal{X}}
\newcommand{\K}{\mathcal{K}}
\newcommand{\F}{\mathcal{F}}
\newcommand{\B}{\mathcal{B}}
\newcommand{\C}{\mathcal{C}}
\newcommand{\G}{\mathcal{G}}
\renewcommand{\S}{\mathcal{S}}
\newcommand{\FS}{\mathcal{S}_F}
\renewcommand{\arc}[1]{\overrightarrow{#1}}

\newcommand{\Id}{\Bbb{I}_n}
\newcommand{\Idk}{\Bbb{I}_k}

\newcommand{\ie}{\textit{i.e.}~}
\newcommand{\eg}{\textit{e.g.}~}

\maketitle

\section{Introduction}

Modular decomposition is a technique at the crossroads of several domains of combinatorics which applies to many discrete structures such as graphs, 2-structures, hypergraphs, set systems and matroids among others. 
As a graph decomposition technique it has been introduced by Gallai~\cite{Gal67} to study the structure of comparability graphs (those graphs whose edge set can be transitively oriented). Roughly speaking a \emph{module} in graph is a subset $M$ of vertices which share the same neighbourhood outside $M$. Galai showed that the family of modules of an undirected graph can be represented by a tree, the \emph{modular decomposition tree}. The notion of module appeared in the litterature as \emph{closed sets}~\cite{Gal67}, \emph{clan}~\cite{EGMS94}, \emph{automonous sets}~\cite{Moh85}, \emph{clumps}~\cite{Bla78}\dots while the modular decomposition is also called \emph{substitution decomposition}~\cite{Moh85b} or \emph{$X$-join decomposition}~\cite{HM79}. See~\cite{MR84} for an early survey on this topic.

There is a large variety of combinatorial applications of modular decomposition. Modules can help proving structural results on graphs  as Galai did for comparability graphs. More generally modular decomposition appears in (but is not limited to) the context of perfect graph theory. Indeed Lov\'asz's proof of the perfect graph theorem~\cite{Lov72} involves cliques modules. Notice also that a number of perfect graph classes can be characterized by properties of their modular decomposition tree: cographs, $P_4$-sparse graphs, permutation graphs, interval graphs\dots~Refer to the books of Golumbic~\cite{Gol80}, Brandst\"adt \emph{et al.}~\cite{BLS99} for graph classes. We should also mention that the modular decomposition tree is useful to solve optimization problems on graphs or other discrete structures (see~\cite{Moh85}). An example of such use is given in the last section.

In the late 70's, the modular decomposition has been independently generalized to \emph{partitive set families}~\cite{CHM81} and to a combinatorial decomposition theory~\cite{CE80} which applies to graphs, matroids and hypergraphs. More recently, the theory of partitive families and its variants had been the foundation of decomposition schemes for various discrete structures among which $2$-structures~\cite{EHR99} and permutations~\cite{UY00,BCMR08}. Beside, based on efficiently representable set families, different graph decompositions had been proposed. The \emph{split decomposition} of~\cite{CE80} relies on a \emph{bipartitive family} on the vertex set. Refer to~\cite{Bui08} for a survey on the recent developments of these techniques.

A good feature of most of these decomposition schemes is that they can be computed in polynomial time.
Indeed, since the early 70's, there have been a number algorithms for computing the modular decomposition of a graph (or for some variants of this problem). The first polynomial algorithm is due to
Cowan, James and Stanton~\cite{CJS72} and runs in $O(n^4)$. Successive improvements are due to Habib and Maurer~\cite{HM79} who proposed a cubic time algorithm, and to M\"uller and Spinrad who designed a quadratic time algorithm. The first two linear time algorithms appeared independently in 1994~\cite{CH94,MS94}. Since then a series of simplified algorithms has been published, some running in linear time~\cite{MS99,TCHP08}, others in almost linear time~\cite{DGM01,MS00,HPV99}. The list is not exhaustive. This line of research yields a series of new interesting algorithmic techniques, which we believe, could be useful in other applications or topics of computer science. The aim of this paper is to survey the algorithmic theory of modular decomposition.

The paper is organized as follows.  The partitive family theory and its application to modular decomposition of graphs is presented in Section 2. As an algorithmic appetizer, Section 3 addresses the special case of totally decomposable graphs, namely the \emph{cographs}, for which a linear time algorithm is known since 1985~\cite{CPS85}. \emph{Partition refinement} is an algorithmic technique that reveals to be really powerful for the modular decomposition problem, but also for other graphs applications (see \eg~\cite{PT87,HPV99}). Section 4 is devoted to partition refinement. Section 5 describes the principle of a series of modular decomposition algorithms developped in the mid 90's. Section 6 explains how the modular decomposition can be efficiently computed via the recent concept of \emph{factoring permutation}~\cite{CHM02}. Let us mention that we do not discuss the recent linear time algorithm of Tedder et al.~\cite{TCHP08}, even though we believe that this last algorithm provides a positive answer to the problem of finding a simple linear time modular decomposition algorithm. Actually the key to Tedder et al.'s algorithm is to merge the ideas developed in Sections 5 and 6. The purpose of this paper is not to enter into the details of all the algorithm techniques but rather to present their main lines. Finally the last section presents three recent applications of the modular decomposition in three different domains of computer science, namely pattern matching, computational biology and parameterized complexity.

\section{Partitive families}

The \emph{modular decomposition theory} has to be understood as a special case of the theory of \emph{partitive family} whose study dates back to the early 80's~\cite{CE80,CHM81}. We briefly present the mains concepts and theorems of the partitive family theory. We then introduce the \emph{modular decomposition} of graphs and discuss its elementary algorithmic aspects. This section ends with a discussion on two important class of graphs: indecomposable graphs (the \emph{prime} graphs) and totally decomposable graphs (known as the \emph{cographs}) 

\subsection{Decomposition theorem of partitive families}
\label{sec:partitive}

The \emph{symmetric difference} between two sets $A$ and $B$ is denoted by $A\vartriangle B=(A\setminus B)\cup (B\setminus A)$. Two subsets $A$ and $B$ of a set $S$ \emph{overlap} if $A\cap B\neq\emptyset$, $A\setminus B\neq\emptyset$ and $B\setminus A\neq\emptyset$, we write $A\overlap B$\index{$\overlap$}. 

\begin{definition} \label{def:partitive}
A family $\mathcal{S}\subseteq 2^S$ of subsets of $S$ is \emph{partitive}\index{partitive family} if:
\vspace{-0.4cm}
\begin{enumerate}
\item $S\in\mathcal{S}$, $\emptyset\notin \mathcal{S}$ and for all $x\in S$, $\{x\}\in \mathcal{S}$;
\item For any pair of subsets $A,B\in\mathcal{S}$ such that  $A\overlap B$:\\
\vspace{-0.4cm}
\begin{enumerate}
\item $A\cap B\in\mathcal{S}$;
\item $A\setminus B\in\mathcal{S}$ and $B\setminus A\in\mathcal{S}$;
\item $A\cup B\in\mathcal{S}$;
\item $A\vartriangle B\in\mathcal{S}$.
\end{enumerate}
\end{enumerate}
\end{definition}

A family is \emph{weakly partitive} \index{weakly partitive family} whenever condition \textit{(2.d)} is not satisfied. 
Unless explicitly mentioned, we will only consider partitive families.

\begin{definition}
An element $F\in \mathcal{S}$ is \emph{strong} if it does not overlap any other element of $\mathcal{S}$. The set of strong elements of $\mathcal{S}$ is denoted $\FS$. 
\end{definition}

Obviously any trivial subset of $\mathcal{S}$, namely $S$ or $\{x\}$ (for $x\in S$), is a strong element. Let us remark that $\FS$ is nested, \ie  the transitive reduction of the inclusion order of  $\FS$ is a tree $T_{\S}$, which we call the \emph{strong element tree} (see Figure~\ref{fig:arbre-part}). It follows that  $|\FS|=O(|S|)$.

\begin{figure}[htp]
\centerline{\scalebox{1}{\includegraphics[width=11cm]{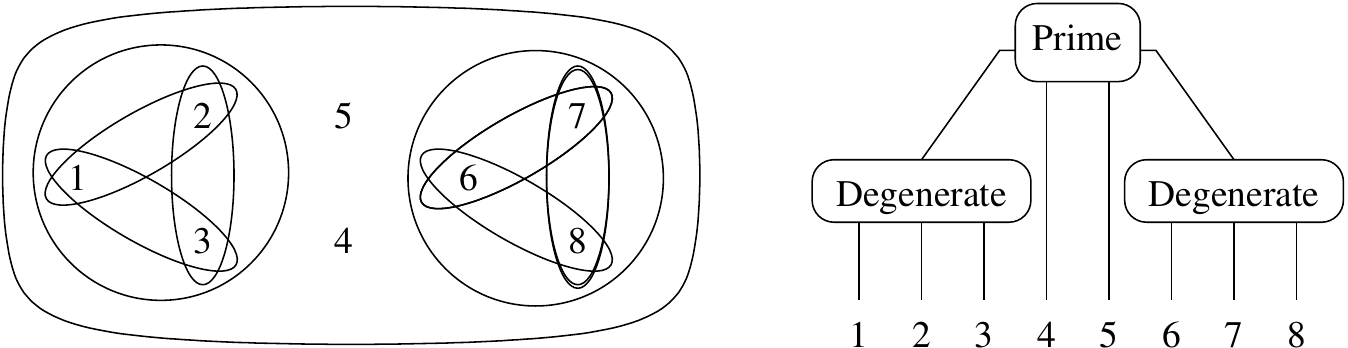}}}
\caption{The inclusion tree of the strong elements of the family $\S=\{\{1,2,3,4,5,6,7,8\},\{1,2,3\},\{6,7,8\},\{1,2\},\{2,3\},\{1,3\},\{6,7\},\{7,8\},\{6,8\},$ $\{1\},\{2\},\{3\},\{4\},\{5\},\{6\},\{7\},\{8\}\}$} \label{fig:arbre-part}
\end{figure}

\begin{definition}
Let $f_1^q,\dots, f_k^q$ be the chidren of a node $q$ of $T_{\S}$, the strong element tree of $\S$. The node $q$ is \emph{degenerate}  if for all non-empty subset $J\subset [1,k]$, $\cup_{j\in J} f_j^q\in\S$. A node is \emph{prime} if for every non-empty subset $J\subset [1,k]$, $\cup_{j\in J} f_j^q\notin\S$.
\end{definition}

It is not difficult to see that any strong element is either prime or degenerate. Moreover the following theorem tells us that the tree $T_{\S}$ is a representation of the family $\S$ and the subfamily of strong elements $\FS$ of $\S$ defines a "basis" of $\S$. 

\begin{theorem}\cite{CHM81} \label{th:partitive}
Let $\S$ be a partitive family on $S$. The subset $A\subseteq S$ belongs to $\S$ if and only if $A$ is strong or there exists a degenerate strong element $A'$ (or a node of $T_{\S}$) such that $A$ is the union of a strict subset of the children of $A'$ in $T_{\S}$.
\end{theorem}

As a consequence, even if a partitive family on a set $S$ can have exponentially many elements, it always admits a representation linear in the size of $S$. Such a representation property  is also known for other families of subsets of a set, such as laminar families, cross-free families~\cite{EG77}\dots as well as for some families of bipartitions of a set, such as splits~\cite{CE80}. Recently, a similar result has been shown for \emph{union-difference} families of subsets of a set, \textit{i.e.} families closed under the union and the difference of its overlapping elements \cite{BH08}. In this latter case, the size of the representation amounts to $O(|S|^2)$. For a detailed study of these aspects, the reader should refer to~\cite{Bui08}.

\subsection{Factoring Permutations}
\label{sec:permutation}

Although the idea of \emph{factoring permutation} implicitly appeared in some early papers (see \eg \cite{HM91,Hsu92,HHS95}), it has only been formalized in ~\cite{CH97, Cap97}. This concept turns out to be central to recent modular decomposition algorithms and other applications.

Let $\sigma$ be a permutation of a set $S$ of size $n$. By $\sigma(x)$, we mean the rank $i$ of $x$ in $\sigma$ and $\sigma^{-1}(i)$ stands for the $i$-th element of $\sigma$.
\index{interval}
A subset $I\subseteq S$ is a \emph{factor} or an \emph{interval} of a permutation $\sigma$ if there exist $i\in[1,n]$ and $j\in[1,n]$ such that $I=\{x \mid x=\sigma^{-1}(k), i\leqslant k\leqslant j\}$. In other words, the elements of $I$ occur consecutively in $\sigma$.

\begin{definition}\cite{Cap97} \index{factoring permutation} \label{def:perm-fact}
Let $\mathcal{S}$ be a (weakly) partitive family of a set $S$ and let $\mathcal{S}_F$ be the strong elements of $\mathcal{S}$. A permutation $\sigma$ of $S$ is \emph{factoring} for $\mathcal{S}$ if for any $F\in\mathcal{S}_F$, $F$ is a factor of $\sigma$.
\end{definition}

For example, $\pi=1~2~3~4~5~6~7~8$, $\pi_1=6~7~8~4~3~1~2~5$ and $\pi_2=8~7~6~1~3~2~4~5$ are three factoring permutations of the family $\mathcal{S}$ depicted in  Figure~\ref{fig:arbre-part}. One can check that, in each of these three permutations, the two non-trivial strong elements of $\mathcal{S}_F$, namely  $\{1,2,3\}\in\mathcal{S}_F$ and $\{6,7,8\}\in\mathcal{S}_F$, are factors.

Given a layout of the strong element tree of a partitive family, a left-to-right enumeration of the leaves results in a factoring permutation. In many cases it is easier to compute a factoring permutation than the strong element tree.
We  explain in Section~\ref{sec:perm-to-tree}  how to obtain the strong element tree from a factoring permutation.

To conclude this brief introduction on factorizing permutation, we state a Lemma which formalizes links between intervals of factoring permuations and partitive families. This Lemma somehow guided the development  of factoring permutation algorithms.

\begin{lemma} \label{lem:fp-perm-partitive}
Let $\sigma$ be a factoring permutation of a partitive family $\S$. Then the set $\mathcal{I}(\S,\sigma)$ of intervals of $\sigma$ which are elements of $\S$ is a weakly partitive family. Moreover the strong elements  of $\mathcal{I}(\S,\sigma)$ and of $\S$ are the same.
\end{lemma}


\subsection{Modules of a graph} 
\label{sec:module}

For the sake of the presentation we only consider undirected, simple and loopless graphs. We use the classical notations (\eg see~\cite{BLS99}). The neighbourhood of a vertex $x$ in a graph $G=(V,E)$ is denoted $N_G(x)$ and its non-neighbourhood $\overline{N}_G(x)$ (subscript $G$ will be omitted when the context is clear). The complementary graph of a graph $G$ is denoted by $\overline{G}$. Given a subset of vertices $X\subseteq V$, $G[X]$ is the subgraph induced by $X$ (any edge in $G$ between two vertices in $X$ belongs to $G[X]$).

Let $M$ be a set of vertices of a graph $G=(V,E)$ and $x$ be a vertex of $V\setminus M$. Vertex $x$ \emph{splits} $M$ (or is a \emph{splitter}\index{splitter} of $M$), if there exist $y\in M$ and $z\in M$ such that $xy\in E$ and $xz\notin E$. If $x$ is not a splitter of $M$, then $M$ is \emph{uniform} or \emph{homogeneous}\index{homogeneous} with respect to $x$.

\begin{definition} \index{module}
Let $G=(V,E)$ be a graph. A set  $M\subseteq V$ of vertices is a  \emph{module} if $M$ is homogeneous with respect to any $x\notin M$ (i.e. $M\subseteq N(x)$ or $M\cap N(x)=\emptyset$).
\end{definition}

\begin{observation} \label{obs:splitter}
Let $S$ be a subset of vertices of a graph $G=(V,E)$. If $S$ has a splitter $x$, then any module of $G$ containing $S$ also contains $x$.
\end{observation}

Aside the singletons and the whole vertex sets, any union of  connected components (or of co-connected components) of a graph are simple examples of modules. Let us also note that a graph may have exponentially many modules. Indeed any subset of a complete graph is a clique. Nevertheless, as we shall see with the following lemma, the family of modules has strong combinatorial properties.

\begin{lemma} \cite{CHM81}\label{lem:mod-partitive}
The family $\M$ of modules of a graph is partitive.
\end{lemma}

The notions of \emph{trivial} and  \emph{strong} module and \emph{degenerate}  are defined according to the terminology of  Section~\ref{sec:partitive}. By Lemma~\ref{lem:mod-partitive}, if $M$ and $M'$ are overlapping modules, then 
$M\setminus M'$, $M'\setminus M$, $M\cap M'$, $M\cup M'$ and $M\vartriangle M'$ are modules of $G$.

Let $M$ and $M'$ be disjoint sets. We say that $M$ and $M'$ are \emph{adjacent} if any vertex of $M$ is adjacent to all the vertices of $M'$ and \emph{non-adjacent} if the vertices of $M$ are non-adjacent to the vertices of $M'$.

\begin{observation} \label{obs:module-adj}
Two disjoint modules are either adjacent or non-adjacent.
\end{observation}

A module $M$ is \emph{maximal} with respect to a set $S$ of vertices, if $M\subset S$ and there is no module $M'$ such that $M\subset M'\subset S$. If the set $S$ is not specified, we shall assume $S=V$.

\begin{definition} \index{modular partition}
Let $\P=\{M_1,\dots, M_k\}$ be a partition of the vertex set of a graph $G=(V,E)$. If for all $i$, $1\leqslant i\leqslant k$, $M_i$ is a module of $G$, then $\P$ is a \emph{modular partition} (or \emph{congruence} partition) of $G$.
\end{definition}

\index{maximal modular partition}
A non-trivial modular partition  $\P=\{M_1,\dots, M_k\}$ which only contains maximal strong modules is a \emph{maximal modular partition}. Notice that each graph has a unique maximal modular partition. If $G$ (resp. $\overline{G}$) is not connected then its (resp. co-connected) connected components are the elements of the maximal modular partition. From Observation~\ref{obs:module-adj}, we can define a \emph{quotient graph} whose vertices are the parts (or modules) belonging to the modular partition $\P$.

\begin{definition} \index{quotient graph}
To a modular partition $\P=\{M_1,\dots, M_k\}$ of a graph $G=(V,E)$, we associate a \emph{quotient graph} $G_{/\P}$, whose vertices are in one-to-one correspondence with the parts of $\P$. Two vertices $v_i$ and $v_j$ of $G_{/\P}$ are adjacent if and only if the corresponding modules $M_i$ and $M_j$ are adjacent in $G$.
\end{definition}

\begin{figure}[htp]
\centerline{\scalebox{1}{\includegraphics[width=11cm]{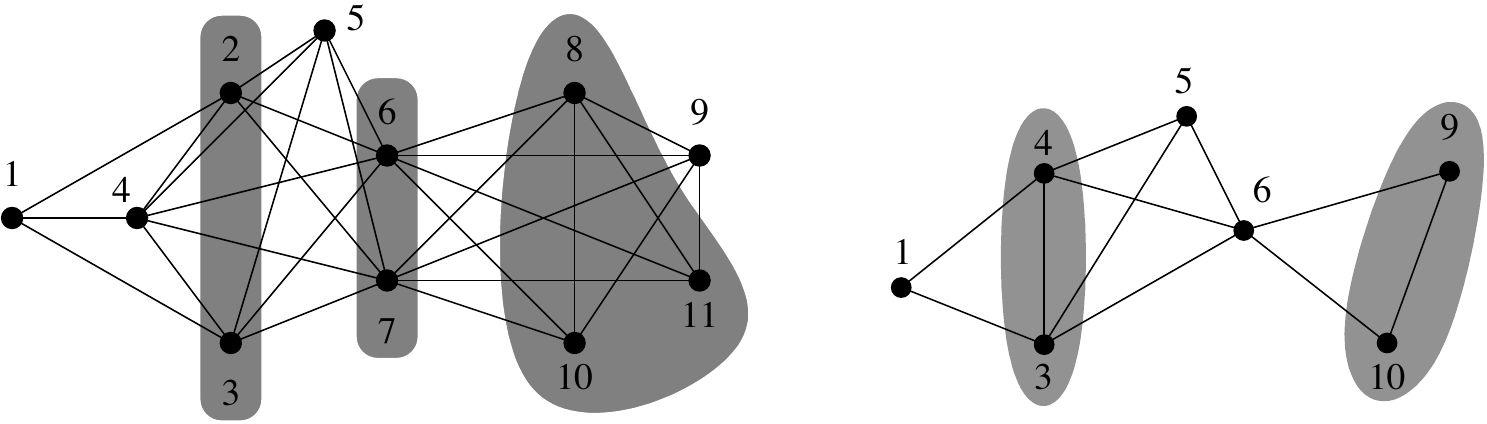}}}
\caption{On the left, the grey sets are modules of the graph $G$. $\mathcal{Q}=\{\{1\},\{2,3\},\{4\},\{5\},\{6,7\},\{9\},\{8,10,11\}\}$ is a modular partition of $G$. The quotient graph $G_{/\mathcal{Q}}$, depicted on the right with a representative vertex for each module of $\mathcal{Q}$, has two non-trivial modules (the sets $\{3,4\}$ and $\{9,10\}$). The maximal modular partition of $G$ is $\P=\{\{1\},\{2,3,4\},\{5\},\{6,7\},\{8,9,10,11\}\}$ and its quotient graph are represented in Figure~\ref{fig:md-tree} (aside the top node of the tree).}
\label{fig:partition-module}
\end{figure}

Let us remark that the quotient graph $G_{/\P}$ with $\P=\{M_1,\dots, M_k\}$ is isomorphic to any subgraph induced by a set $V'\subseteq V$ such that $\forall i\in[1,k]$, $|M_i\cap V'|=1$.
The \emph{representative graph} of a module $M$ is the quotient graph $G[M]_{/\P}$ where $\P$ is the maximal modular partition of $G[M]$: it is thereby the subgraph induced by a set containing a unique -- \emph{representative} -- vertex per maximal strong module of $G[M]$. See Figure~\ref{fig:partition-module}.
By extension, for a module $M$, we denote by $G_{/M}$ the graph quotiented by the modular partition $\{M\}\cup\{\{x\}\mid x\notin M\}$.
 
Before we state the \emph{modular decomposition theorem} (Theorem~\ref{th:dec-module}), let us present two more properties of modular partitions and quotient graphs which are central to efficient modular decomposition algorithms (see Section~\ref{sec:algo}).

\begin{lemma} \cite{Moh85} \label{lem:partition}
Let $\P$ be a modular partition of a graph $G=(V,E)$. Then $\X\subseteq\P$ is a module of $G_{/\P}$ iff $\bigcup_{M\in\X} M$ is a module of $G$.
\end{lemma}

Lemma~\ref{lem:partition} is illustrated on Figure~\ref{fig:partition-module}: for example, the set $\{2,3,4\}$ is a module of $G$, it is the union of modules $\{2,3\}$ and $\{4\}$ (which representative vertices are respectively $3$ and $4$ in $G_{/\mathcal{Q}}$) which belongs to partition $\mathcal{Q}$. It can be strengthened in order to observe the correspondance between the strong modules of $G$ and those of $G_{/\P}$.

\begin{lemma} \label{lem:partition-fort}
Let $\P$ be a modular partition of a graph $G=(V,E)$. Then $\X\subset\P$ is a non-trivial \emph{strong} module of $G_{/\P}$ iff $\bigcup_{M\in\X} M$ is a non trivial \emph{strong} module of $G$.
\end{lemma}

The inclusion tree of the strong modules of $G$, denoted $MD(G)$, entirely represents the graph if the representative graph of each strong module is attached to each of its nodes (see Figure~\ref{fig:md-tree}). Indeed any adjacency of $G$ can be retrieved from $MD(G)$. Let $x$ and $y$ be two vertices of $G$ and let $G_N$ be the representative graph of node $N$, their least common ancestor. Then $x$ and $y$ are adjacent in $G$ if and only if their representative vertices in $G_N$ are adjacent.

\begin{figure}[htp]
\centerline{\scalebox{1}{\includegraphics[width=11cm]{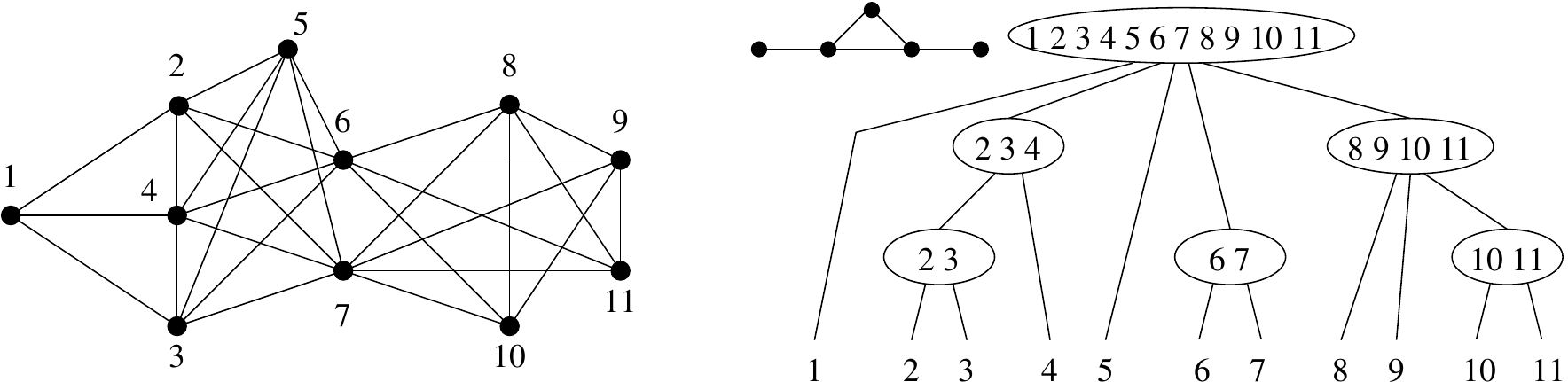}}}
\caption{The inclusion tree $MD(G)$ of the strong modules of $G$. The representative graph associated to the root is $G_{/\P}$ with $\P=\{\{1\},\{2,3,4\},\{5\},\{6,7\},\{8,9,10,11\}\}$, the parts of which correspond to the children of the root.}
\label{fig:md-tree}
\end{figure}

Let us recall that a graph is \emph{prime} if it only contains trivial modules.

\begin{theorem}[Modular decomposition theorem] \cite{Gal67,CHM81} \label{th:dec-module}\\
For any graph $G=(V,E)$, one of the following three conditions is satisfied:
\vspace{-0.4cm}
\begin{enumerate}
\item $G$ is not connected;
\item $\overline{G}$ is not connected;
\item $G$ and $\overline{G}$ are connected and the quotient graph $G_{/\P}$, with $\P$ the maximal modular partition of $G$, is a prime graph.
\end{enumerate}
\end{theorem}

What does the modular decomposition theorem say is twofold. First, the quotient graphs associated with the nodes of the inclusion tree $MD(G)$ of the strong modules are of three types: an \emph{independent set} if $G$ is not connected (the node is labelled \emph{parallel}); a \emph{clique} (complete graph) if $\overline{G}$ is not connected (the node is labelled \emph{series});  a prime graph otherwise. It also follows that $MD(G)$ is unique and does not contain two consecutive series nodes nor two consecutive parallel nodes. Parallel and series nodes of $MD(G)$ are also called \emph{degenerate} nodes.

The tree $MD(G)$ is called the \emph{modular decomposition tree}. Theorem~\ref{th:dec-module} yields a natural polynomial time recursive algorithm to compute $MD(G)$: 1) compute the maximal modular partition $\P$ of $G$; 2) label the root node according to the parallel, series or prime type of $G$; 3) for each module $M$ of $\P$, compute $MD(G[M])$ and attach it to the root node. A subproblem central to the computation of $MD(G)$ is to compute the maximal modular partition, a task which can be avoided if a non-trivial module $M$ is identified. This yields another natural algorithm scheme: by Lemma~\ref{lem:partition} and Lemma~\ref{lem:partition-fort}, it suffices to recursively compute $MD(G[M])$ and $MD(G_{/M})$, and then to paste $MD(G[M])$ on the leaf of $MD(G_{/M})$ corresponding to the representative vertex of $M$. As suggested by Cowan et al.~\cite{CJS72}, a naive way to compute a non-trivial module is to follow the definition of module and Observation~\ref{obs:splitter}. Assume the graph $G$ contains a non-trivial module $M$. Then $M$ contains a pair of vertices $\{x,y\}$ and as a module is closed under adding splitters. Such an algorithm would find a non-trivial module, if any, in time $O(n^2(n+m))$. We should note that for some generalizations of the modular decomposition, no better algorithm than this "closure by splitter" approach is known (see e.g.~\cite{BHLM09}).

Before we present some structural properties of prime and totally decomposable graphs, let us introduce some notations and briefly discuss the composition view of the theory of modules in graphs.

\begin{notation}
For a node  $p$ of $MD(G)$, its corresponding strong module  is denoted by $M(p)$ (or $P$). In fact $M(p)$ is the union of all singletons which are leaves of the subtree of $M(p)$ rooted in $p$.\\
The minimal strong module containing two vertices $x$ and $y$ is  denoted by $m(x,y)$, while  the maximal strong module containing $x$ but not $y$, for any two different vertices  $x, y$ of $G$, is denoted by $M(x,\overline{y})$.
\end{notation}

The \emph{substitution}\index{substitution} operation is the reverse of the quotient operation. It consists of replacing a vertex $x$ of $G$ by a graph $H=(V',E')$ while preserving the neighourhood. The resulting graph is:
$$G_{x\rightarrow H}=((V\setminus\{x\})\cup V', (E\setminus\{xy\in E\})\cup E'\cup \{yz:xy\in E\mbox{ et } z\in V'\})$$ 

The \emph{parallel composition} \index{parallel composition} or \emph{disjoint union} of $k$ connected graphs $G_1,\dots G_k$ defines a graph whose connected components are the graphs $G_1,\dots, G_k$. This composition operation is usually denoted $G_1\oplus\dots \oplus G_k$.

The \emph{series composition} \index{series composition} of $k$ co-connected graphs $G_1,\dots, G_k$ defines a graph whose co-connected components are the graphs $G_1,\dots, G_k$ (for any pair $x,y$ of vertices belonging to different graphs $G_i$ and $G_j$, the edge $xy$ has been added). The series composition is generally  denoted  $G_1\otimes\dots \otimes G_k$.

These three operations are classical graph operations that have been widely used in various contexts among which the clique-width theory~\cite{CER93}.

\subsection{Prime graphs}

The structure of prime graphs has been extensively studied (\eg see~\cite{ER90,ST93b,CI98}). For example, it is easy to check that the smallest prime graph is the $P_4$\index{$P_4$},
the path on $4$ vertices (see Figure~\ref{fig:p4}). As witnessed by the following result, $P_4$'s play an important role in the structure of prime graphs.

\begin{lemma}\cite{CI98} \label{lem:p4}
Let $G$, with $|G| \geq 4$, be a prime graph. Then any vertex, but at most one, is contained in an induced $P_4$. A vertex not contained in any $P_4$ is called the "\emph{nose of the bull}" (see Figure~\ref{fig:p4}).
\end{lemma}

\begin{figure}[htp]
\centerline{\scalebox{1}{\includegraphics[width=8cm]{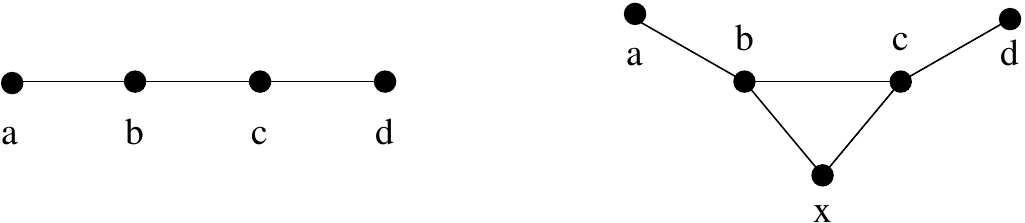}}}
\caption{
The vertices $a,b,c,d$ form a $P_4$ whose \emph{extremities} are $a$ and $d$, and midpoints $b$ and $c$. The graph on the right is the \emph{bull} whose\emph{"nose"}  is vertex $x$.}
\label{fig:p4}
\end{figure}

The next property shows that one can always remove one or two vertices from a large enough prime graph to obtain a new prime graph.

\begin{lemma} \cite{ER90,ST93b}
Let $G=(V,E)$ be a prime graph with at least $5$ vertices. Then there exists a subset of vertices $X$ such that $|V|-2\leqslant |X|\leqslant |V|-1$ and $G[X]$ is prime.
\end{lemma}

Jamison and Olariu proposed an extension of Theorem~\ref{th:dec-module} by considering the structure of prime graphs~\cite{JO95}.  A subset $C$ of vertices of a graph $G=(V,E)$ is \emph{$P$-connected} if for any bipartition $\{A,B\}$ of $C$, there is an induced  $P_4$ intersecting both $A$ and $B$. For example the bull is not $P$-connected (consider the vertex partition $\{\{x\},\{a,b,c,d\}\}$). A $P$-connected component is a maximal $P$-connected set of vertices. The set of $P$-connected components defines a partition of the vertices. A $P$-connected component $H$ is \emph{separable} if there is a bipartition $(H_1,H_2)$ of $H$ such that for any $P_4$ intersecting $H_1$ and $H_2$, the extremities are in $H_1$ and the mid-vertices in $H_2$.

\begin{theorem} \cite{JO95} \label{th:p-connected}
Let  $G=(V,E)$ be a connected graph such that $\overline{G}$ is connected. then $G$ is either $P$-connected or there exists a unique $P$-connected component $H$ which is separable in $(H_1,H_2)$ such that for any vertex $x\notin H$, $H_1\subseteq N(x)$ and $H_2\cap N(x)=\emptyset$.
\end{theorem}

A hierarchy of graph families have been proposed based on the above Theorem \ref{th:p-connected} by restricting the number of induced $P_4$'s in small subgraphs (or equivalently by restricting  the structure of prime graphs). For example, $P_4$-sparse graphs are defined as the graphs for which there is at most one $P_4$ in any induced subgraph on $5$ vertices~\cite{JO92,JO92b}. Let us also mention the the $P_4$-reducible graphs~\cite{JO95}. See~\cite{BLS99} for a complete presentation of these graph families.

\subsection{Totally decomposable graphs}
\label{sub:cographs}

\index{cograph} 
A graph is \emph{totally decomposable} if any induced subgraph of size at least $4$ has a non-trivial module. As any prime graph contains a $P_4$, it follows from Theorem~\ref{th:partitive} that any node of the modular decomposition tree $MD(G)$  of a totally decomposable graph $G$ is degenerate.

The family $\mathcal{F}$ of totally decomposable graphs is  natural and arose in many different contexts (see  \cite{Sum73,CLS81,CPS85} for references) even recently (see \cite{BBCP05,BRV07}) as any graph of $\F$ can be obtained by a sequence of disjoint and series compositions starting from single vertex graph. Let us remark that if $G$ is totally decomposable then also is its complement. The family of totally decomposable graphs is also known as the \emph{cographs}  for  \emph{complement reducible graphs}~\cite{CLS81,Sum73}. From definition, the cograph family is \emph{hereditary} (any induced subgraph of a cograph is a cograph). It also has a very simple forbidden subgraph characterization.

\begin{theorem} \cite{Sum73} \label{theo:p4free}
The cographs are exactly the $P_4$-free graphs.
\end{theorem}

\begin{figure}[htp]
\centerline{\includegraphics[width=11cm]{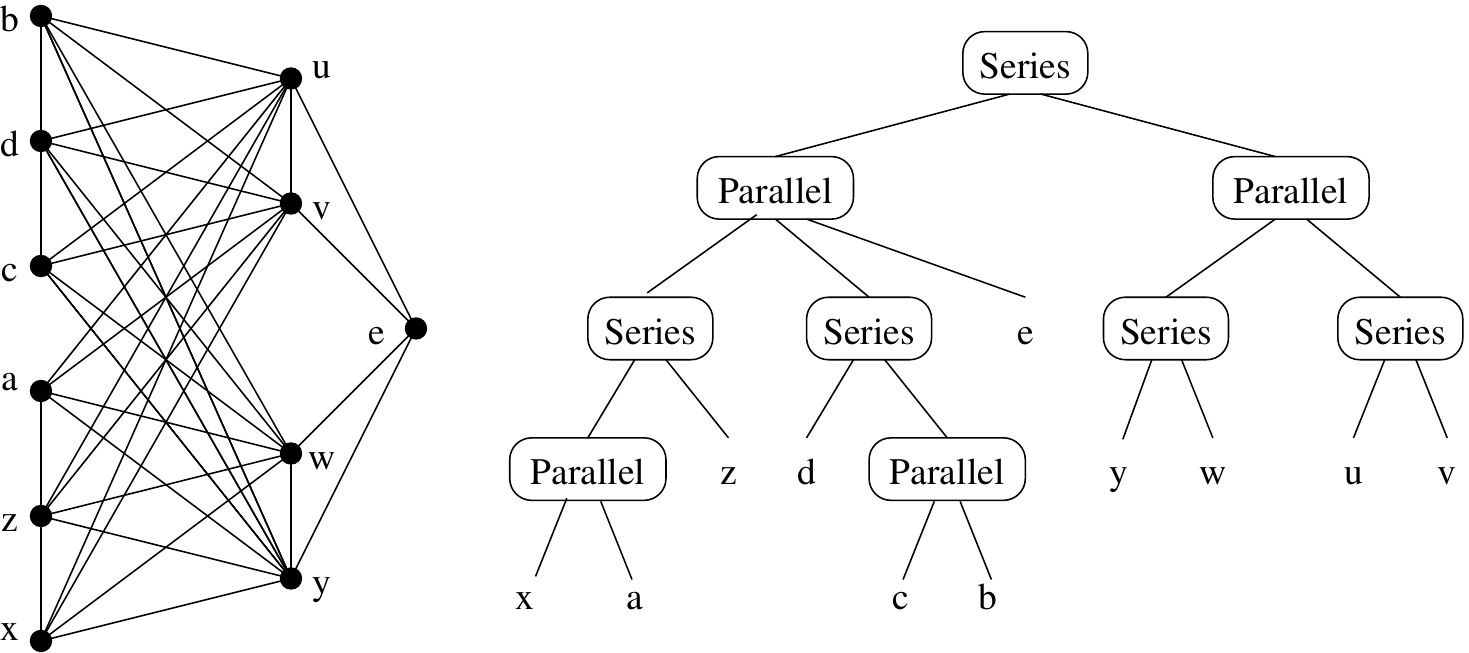}}
\caption{A cograph and its modular decomposition tree (also called cotree).\label{fig:cographe}}
\end{figure}

The following lemma states classical properties of cographs whose proofs (left to the reader) are good exercises to understand the structure of cographs.

\begin{lemma} \label{lem:basic-cograph}
Let $x,y$ and $v$ be vertices of a cograph $G=(V,E)$.
\vspace{-0.4cm}
\begin{enumerate}
\item If $xv\in E$, $yv\notin E$ and $xy\in E$, then $m(v,y)\subseteq M(v,\overline{x})$
\item If $xv\in E$, $yv\in E$ and $xy\notin E$, then $M(v,\overline{x})=M(v,\overline{y})$ and $m(v,x)=m(v,y)$
\end{enumerate}
\end{lemma}

Using Theorem~\ref{theo:p4free}, one can propose a naive cograph recognition algorithm by searching for an induced  $P_{4}$. But so far, most of the linear time cograph recognition algorithms construct the modular decomposition tree and exhibit a $P_{4}$ in case of failure.

The first linear time cograph recognition algorithm was proposed in 1985 by Corneil, Perl and Stewart~\cite{CPS85}. It incrementally constructs the modular decomposition tree, also called \emph{cotree}
\index{cotree} when restricted to cographs, as long as the graph induced by the processed vertices is a cograph. Even if alternative recognition algorithms have recently been proposed~\cite{Dah95,HP05,BCHP03}, the seminal algorithm of~\cite{CPS85} is a corner stone in the algorithmic of the modular decomposition and turns out to have a large impact even for other decomposition technics (e.g. for the split decomposition~\cite{GP07}). We present Corneil et al's algorithm in Section~\ref{sec:cograph}.

\subsection{Bibliographic notes}

The seminal paper on modular decomposition of graphs is probably Gallai's one~\cite{Gal67} on transitive orientation. Up to our knowledge, the only survey paper is due to M\"ohring and Radermacher~\cite{MR84}. More recently, Ehrenfeucht, Harju and Rozenberg~\cite{EHR99} published a book on the decomposition of $2$-structures (a generalization of graphs) which presents  the modular decomposition in a more general framework. In its PhD thesis~\cite{Bui08}, Bui Xuan proposes a survey as well as original results on the representation of set families. Many graph families are well-structured with respect to the modular decomposition, \textit{e.g.} comparability graphs, permutation graphs, cographs\dots For these aspects, the reader should refer to the books of Golumbic~\cite{Gol80} and more recently~\cite{BLS99,Spi03}. The algorithmic aspects are particularly developed in \cite{Gol80,Spi03}.

We saw that the family of modules in a graph is partitive. If we move to directed graphs, then we obtain
a weakly partitive family. The related decomposition of bipartite graph into bi-modules also yields a weakly parititive family \cite{FHMV04}. In order to formalize split decomposition \cite{CE80}, bipartitive families have been introduced \cite{CE80,Cun82}. For a recent survey on all kind of variations on the modular decomposition, the reader should refer to~\cite{Bui08}.

\section{Cographs recognition algorithms as an appetizer}
\label{sec:cograph}

We first study in detail the  Corneil, Pearl and Stewart's algorithm ~\cite{CPS85}.   If the input graph  is a cograph, this vertex-incremental algorithm builds the cotree by adding the vertices one by one in an arbitrary order. Then, we sketch how the cotree of a cograph can be updated under edge modification, a result is due to Shamir and Sharan~\cite{SS04}.

\subsection{Adding a vertex to a cograph}

Consider the following subproblem: given a cograph $G=(V,E)$ together with its cotree $MD(G)$, a vertex $x$ and a subset of vertices $S\subseteq V$, test whether the graph $G+(x,S)=(V\cup\{x\},E\cup\{xy\mid y\in S\})$ is a cograph and if so ouput the cotree $MD(G+x)$. Corneil \emph{et al}'s~\cite{CPS85} showed that whether $G+x$ is a cograph or not can be characterized by a labelling of the nodes of the cotree $MD(G)$. A node $p$ receives the label: \emph{empty}, if the corresponding module $M(p)$ does not intersect $S$; \emph{adjacent} if $M(p)\subseteq S$; and \emph{mixed} otherwise. Remark that by definition any child of a node labelled adjacent (resp. empty) is also labelled adjacent (resp. empty).

\begin{lemma}\cite{CPS85}\label{lem:sommet-cographe}
Let $G$ be a cograph, $x$ a vertex of $V$ and $S\subseteq V$. The graph $G+(x,S)$ is a cograph iff
\begin{enumerate}
\item either none of the nodes of the cotree $MD(G)$ is mixed;
\item or the set of mixed nodes induces a path $\pi$ from the root of $MD(G)$ to some node $p$ and 
\begin{enumerate}
\item the children of the series nodes of $\pi$ different than $p$ are all adjacent; 
\item the children of the parallel nodes of $\pi$ different than $p$ are all empty.
\end{enumerate}
\end{enumerate}
\end{lemma}

The main idea expressed by the conditions of Lemma~\ref{lem:sommet-cographe} is that the modifications of the cotree implied by the insertion of vertex $x$ are localized in the subtree of $MD(G)$ rooted at node $p$. Indeed any module disjoint from $M(p)$ is not affected by $x$'s insertion (the corresponding nodes are labelled empty or adjacent). In a sense, node $p$ should be considered as the insertion node. The cotree updates only depend on node $p$ (\eg whether it is mixed or adjacent). An example is depicted in Figure~\ref{fig:sommet-cographe}.

\begin{figure}[htp]
\centerline{\scalebox{1}{\includegraphics[width=10cm]{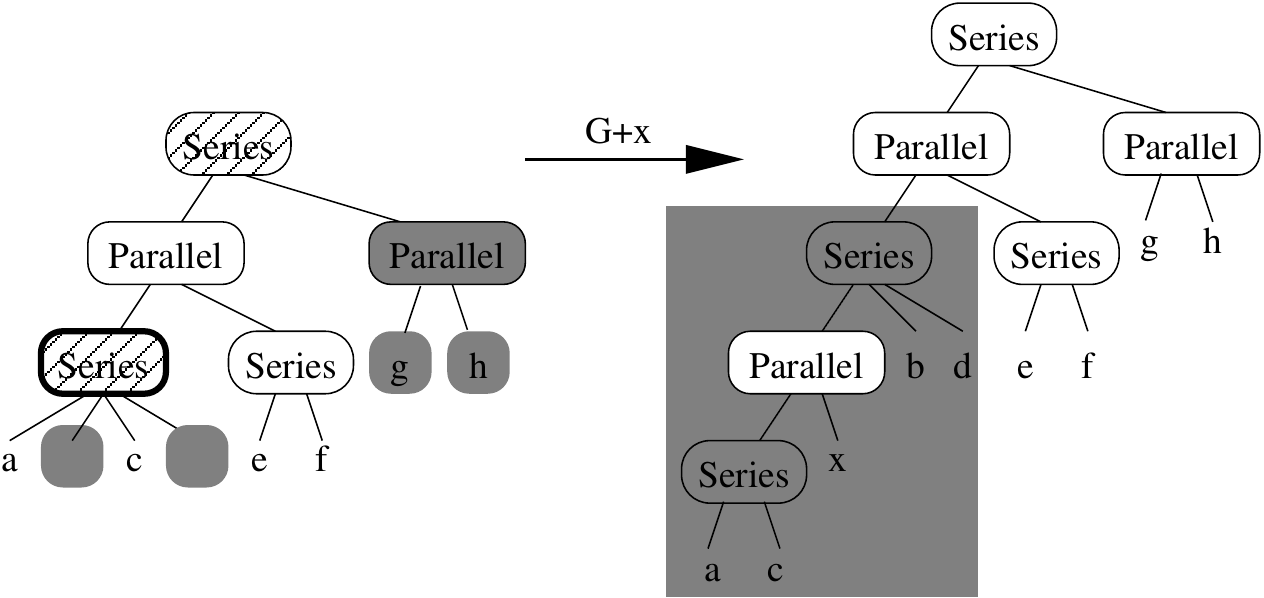}}}
\caption{Insertion of the vertex $x$ adjacent to $S=\{b,d,g,h\}$. Grey nodes are the adjacent labelled nodes and dashed nodes are the mixed nodes. The insertion node $p$ is the bold series node (father of $a$, $b$, $c$, $d$).}
\label{fig:sommet-cographe}
\end{figure}

The algorithm first labels the cotree in a bottom-up manner. The leaves corresponding to vertices of $S$ are labelled adjacent. A node labelled adjacent forwards a partial mark to its father. When a node have received a mark from each of its children, it is labelled adjacent. At the end of this process the empty node have never been searched, while the partially marked nodes corresponds, if $G+x$ is a cograph, to the parallel nodes of the path $\pi$ from the insertion node to the root of $MD(G)$. It is not difficult to see that the number of the marked nodes is linear in the size of $S$ meaning that the labelling process runs in time $O(|S|)$. Testing the condition of the above lemma can be done within the same complexity as well.

\begin{theorem}~\cite{CPS85}
The family of cographs can be recognized in linear time.
\end{theorem}

\subsection{Edge modification algorithms for cographs}

Let us now turn to the edge modification problem which consists in updating the cotree of a cograph $G$ under an edge insertion or deletion. Since the cotree of a cograph can be obtained from the cotree of its complement by flipping the parallel and the series nodes, deleting or inserting an edge in a cograph are equivalent problems.

\begin{lemma} \cite{SS04} \label{lem:arete-cographe}
Let $x$ and $y$ be two non-adjacent vertices of a cograph $G=(V,E)$. Then $G+xy=(V,E\cup\{xy\})$ is a cograph iff
$x$ is a child of $m(x,y)$ and $M(y,\overline{x})\subseteq N(x)$.
\end{lemma}

Let us sketch the argument proof. As $xy\notin E$, the module $m(x,y)$ is represented by a parallel node. Assume the conditions of Lemma~\ref{lem:arete-cographe} do not hold. Then the path in the cotree from $m(x,y)$ to $x$ (resp. $y$) contains a series nodes $p_x$ (resp. $p_y$) which is the least common ancestor of $x$ (resp. $y$) and some leaf $u_x$ (resp. $u_y$). Then the vertices $\{u_x,x,y,u_y\}$ induces a $P_4$ in the graph $G+xy$.

\begin{figure}[htp]
\centerline{\scalebox{1}{\includegraphics[width=9cm]{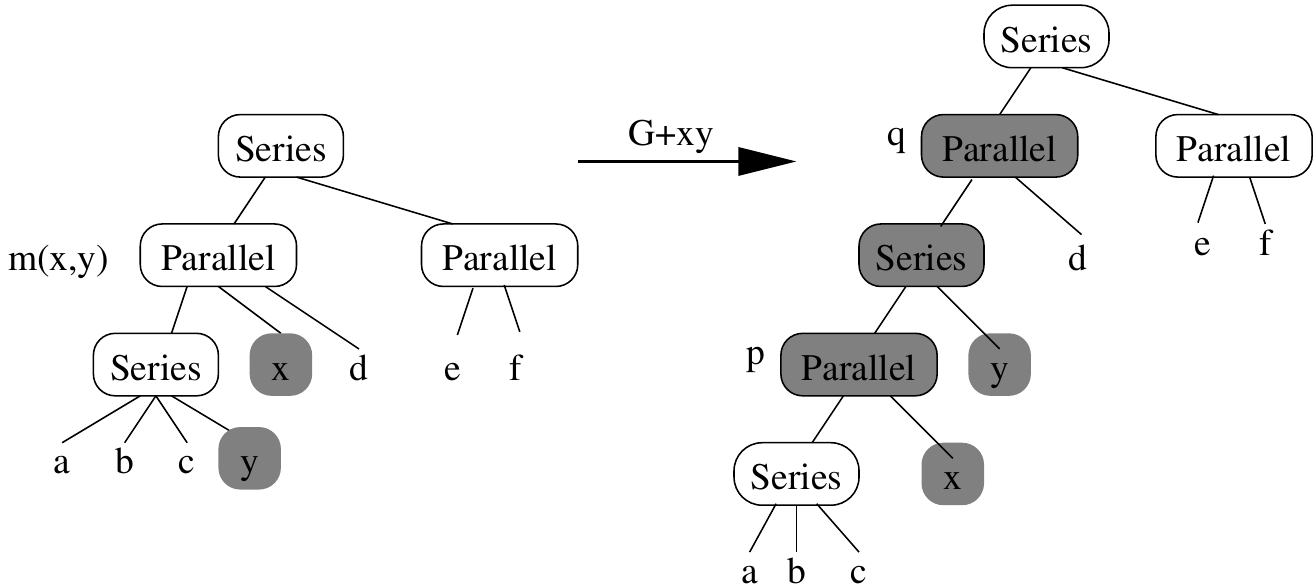}}}
\caption{Update of the cotree to insert the edge $xy$ in a cograph. The node $m(x,y)$ is split into two parallel nodes, say $p$ and $q$, one being the father of $x$, the another the father of the other children of $m(x,y)$. Then leaf $y$ is extracted from the cotree and attached to a new series node inserted between nodes $p$ and $q$.}
\label{fig:arete-cographe}
\end{figure}

It follows from Lemma~\ref{lem:arete-cographe} that as long as the modified graph remains a cograph, the modifications in the cotree are local and can be done in constant time. From results presented in this section, we otbain that:

\begin{theorem} \cite{SS04,CPS85}
There exists an algorithm maintaining the modular decomposition tree of a cograph which runs in time
$O(d)$ per modification (edge or vertex insertion and deletion), where $d$ is the number involved in the modification.
\end{theorem}

Such an algorithm is known in the litterature as a \emph{fully-dynamic} algorithm.

\subsection{Bibliographic notes}
\label{sub:incremental-notes}

In the late 80's, M{\"u}ller and Spinrad generalized Corneil et al's algorithm to the first quadratic modular decomposition algorithm of graphs~\cite{MS89}. Their algorithm is also incremental, but unlike in Corneil et al's algorithm, the whole graph has to be known at the beginning of the algorithm. This restriction is required for the sake of adjacency tests. 

Concerning the cograph recognition problem, new algorithms also appeared recently. Habib and Paul~\cite{HP05} proposed a partition refinement based algorithm (see Section~\ref{sec:partition}) and Bretscher \emph{et al}~\cite{BCHP08} discovered a simple Lexicographic Breadth First Search~\cite{RTL76} based algorithm.

Aside the two cograph algorithmic results presented above, fully-dynamic algorithms have recently been proposed to maintain a representation based on the modular decomposition tree under vertex and edge modifications for various graph classes: permutation graphs~\cite{CP06}, interval graphs~\cite{Cre09,Iba09}\dots The fully-dynamic representation problem has also been solved for other families of graphs, \eg proper interval graphs~\cite{HSS01}, using other decomposition schemes.

Beside, Corneil \emph{et al}'s algorithm has been generalized to the split decomposition~\cite{CE80} to obtain an optimal fully dynamic algorithm for the distance hereditary graphs recognition problem~\cite{GP07}. More recently by the same technique, Gioan \emph{et al.} derived an almost linear time split decomposition algorithm~\cite{GPTC09b} and the first subquadratic circle graph recognition algorithm~\cite{GPTC09a}.

\section{Partition refinement}
\label{sec:partition}

Partition refinement, as an algorithmic technique, has been used in a number of problems, the first of which is probably the deterministic automata minimization~\cite{Hop71}. Paigue and Tarjan~\cite{PT87} wrote a synthesis paper on this technique. Since then, the number of problems solved by partition refinement keeps increasing: interval graph recognition~\cite{HPV99} and completion~\cite{RST08}, transitive orientation, consecutive ones property for boolean matrices ~\cite{HMPV00} are example among others. As we will see, this technique turns out to be a powerful and simple algorithmic paradigm that plays an important role in the context of modular decomposition.

We first present the data-structure and the elementary operation, namely the \emph{refine} operation, of the partition refinement technique. Then, we illustrate this technique with an algorithm that computes a modular partition of a graph. Let us mention that this algorithm really follows the lines of Hopcroft's deterministic automaton minimization algorithm~\cite{Hop71}.

\subsection{Data-structures and algorithmic scheme}

Let $\mathcal{P}$ and $\mathcal{P}'$ be two partitions of the same set $V$. The partition $\mathcal{P}$ is \emph{smaller} than $\mathcal{P}'$, denoted $\mathcal{P}\triangleleft \mathcal{P}'$, if $\mathcal{P}\neq\mathcal{P}'$ and any part of $\mathcal{P}$ is a subset of some part of $\mathcal{P}'$. The partition $\mathcal{P}$ is \emph{stable} with respect to a set $S$ if none of the parts of $\mathcal{P}$ overlaps $S$.

Partition refinement consists of repeating, as long as needed, the operation described in Algorithm~\ref{alg:affinage}. The initial partition and the sequence of \emph{pivot sets} used in the successive refinement steps have a large impact on the whole complexity of the algorithm. Partitioning the vertex set of a graph with respect to the neighbourhood of some vertex is a common operation in graph algorithms. Indeed in our examples, all pivot sets considered correspond to the neighbourhood of some vertex.

\medskip
\begin{algorithm2e}[h]
\caption{\emph{Refine($\mathcal{P}$, $S$)} \label{alg:affinage}}
\KwIn{A partition $\mathcal{P}$ of a set $V$ and a subset $S\subseteq V$, called \emph{pivot set}}
\KwOut{The coarsest partition refining $\mathcal{P}$ and stable for $S$}
\Begin{
\ForEach{part $\mathcal{X}\in\mathcal{P}$}{
	\lIf{$\mathcal{X}\cap S\neq\emptyset$ and $\mathcal{X}\cap S\neq\mathcal{X}$}{
		replace $\mathcal{X}$ by $\mathcal{X}\cap S$ and $\mathcal{X}\setminus S$\;
		}
	}
}
\end{algorithm2e}

\medskip
Let us briefly describe a very  useful data-structure, namely the \emph{standard partition data structure} (see Figure~\ref{fig:affiner}). The elements of the set $V$ to be partitioned are stored in a doubly linked list. Each element of $V$ is assigned a pointer towards the part it belongs to. The elements of a part $\mathcal{X}$ remains consecutive in the doubly linked list (they form an interval). So that each part maintains a pointer towards its first and its last element in the list.

\begin{figure}[h]
\centerline{\scalebox{1}{\includegraphics[width=10cm]{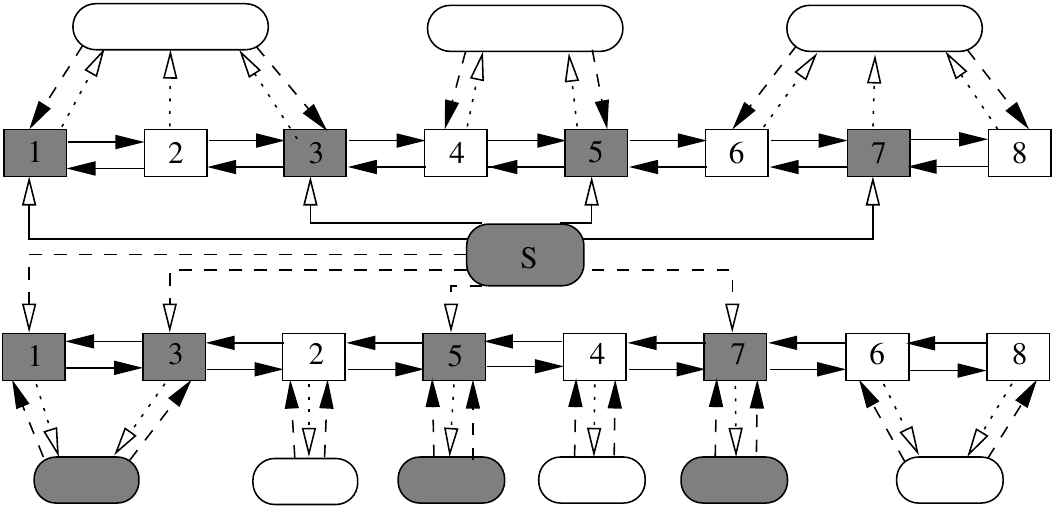}}}
\caption{\emph{$\mathcal{P}'=$Refine$(\mathcal{P},S)$}.\label{fig:affiner}}
\end{figure}

\begin{notation}
The data-structure implicitly represents an \emph{ordered partition}\index{odered partition}: the parts are totally ordered. Depending of the application, this aspect may or may not be important. In order to distinguish the two different cases, an ordered partition will be denoted by $\mathcal{P}=[\mathcal{X}_1,\dots ,\mathcal{X}_k]$ while a non-ordered partition will be denoted by $\mathcal{P}=\{\mathcal{X}_1,\dots ,\mathcal{X}_k\}$. 
\end{notation}

Given a subset $S\subseteq V$, using this standard partition data structure, one can build a list $L$ containing the parts of $\mathcal{P}$ intersecting $S$, such that in each of these parts the elements of $S$ occur first.
Then using $L$, one can split every part into $\mathcal{X}\cap S$  and $\mathcal{X} \setminus S$. A careful complexity analysis shows the following result:

\begin{lemma} \label{lem:affiner}
The time complexity of the operation \emph{Refine($\mathcal{P}$, $S$)} is $O(|S|)$.
\end{lemma}


We conclude this brief introduction by a few remarks. Refining a partition by a subset $S$ or its complement $\overline{S}=V\setminus S$ are equivalent operations: \emph{Refine($\mathcal{X}$, $S$)}= \emph{Refine($\mathcal{X}\cap S$, $V\setminus S$)}. It is thereby possible to deal with the complement of the input graph without explicitly storing its edge set. 
Partition refinement is usually used either to compute a total ordering of the vertices (\eg LexBFS) or the equivalence classes of some equivalence relations (\eg maximal set of twin vertices). McConnell and Spinrad~\cite{MS00} showed how to augment the data-structure in order to extract within the same complexity, at each refinement step, the edges incident to vertices belonging to different parts. This operation is useful to efficiently compute the quotient graph associated to a modular partition. For a more detailed presentation of partition refinement refer to ~\cite{PT87,HPV98,HPV99,HMPV00}.

Of course many variations of the standard partition data structure have been introduced, as for example changing the doubly linked list into an array of size $|V|$.
A further requirement can be that the elements of every part $\mathcal{X}$ of $\mathcal{P}$ are maintained sorted according to a given an initial ordering $\tau$ of $V$. This can be done within the same complexity and is very useful for example when dealing with LexBFS multi-sweep algorithms. The ordering given by some previous LexBFS can be used as a tie-break rule for another LexBFS~\cite{Cor04, Cor04b, BCHP08}.

\subsection{Hopcroft's rule and computation of a modular partition}

Partition refinement is the right tool to compute a modular partition, an important subproblem towards efficient modular decomposition algorithms. In this section, we focus on the problem of computing the \emph{coarsest modular partition} (see Definition~\ref{def:pmm}) of a given vertex partition. The algorithm we present runs in time $O(n+m\log n)$ and is based on the \emph{Hopcroft's rule} which is used in various simple quasi-linear time modular decomposition algorithms.

\begin{definition} \label{def:pmm}
Let $\mathcal{P}$ be a partition of the vertices of a graph $G=(V,E)$. The \emph{coarsest modular partition of $G$ with respect to $\mathcal{P}$} is the largest modular partition $\mathcal{Q}$ such that $\mathcal{Q}\triangleleft \mathcal{P}$.
\end{definition}

The main idea of the algorithm is the following: as long as there is a part $\mathcal{X}$ which is not uniform for some vertex $x\notin\mathcal{X}$, the current partition $\mathcal{P}$ is refined with the neighbourhood $N(x)$. When the algorithm ends, all the parts are modules. Finding, at each step, a vertex $x$ whose neighbourhood strictly refines the partition $\mathcal{P}$, is the usual barrier to linear time complexity. However, using the so-called \emph{Hopcroft's rule}, one get a fairly simple solution that uses the neighbourhood of each vertex at most $\log n$ times.

\begin{lemma}
Let $\mathcal{P}$ be a partition of the vertices of a graph $G=(V,E)$ and $x$ be a vertex of some part $\mathcal{X}$. If  $\mathcal{P}$ is stable with respect to $N(y)$, $\forall y\notin\mathcal{X}$, then $\mathcal{X}$ is a module of $G$ and the partition $\mathcal{Q}=Refine(\mathcal{P},N(x))$ is stable with respect to $N(x')$, $\forall x'\in\mathcal{X}$.
\end{lemma}

The above lemma (which is a direct consequence of the definition of module) shows that using as pivots the vertices of all the parts of $\mathcal{P}$ but one, say $\mathcal{Z}$, plus one vertex $z$ of $\mathcal{Z}$ is enough. For complexity issues, the avoided part $\mathcal{Z}$ has to be chosen as the largest part of $\mathcal{P}$. Similarly, once a part $\mathcal{X}$ has been split, the process continues recursively on the subgraph induced by $\mathcal{X}$ and the resulting largest subpart can be avoided (meaning that only one of its vertices has to be used as pivot). This  \textit{"avoid the largest part"} technique is known as the \emph{Hopcroft's rule} and has been first proposed in the deterministic automata minimization algorithm~\cite{Hop71}.

\begin{algorithm2e}[h]
\caption{Modular Partition \label{alg:part-modulaire}}
\small{
\KwIn{A partition $\mathcal{P}$ of the vertex set  $V$ of a graph $G$}
\KwOut{The coarsest modular partition $\mathcal{Q}$ smaller than $\mathcal{P}$}
\Begin{
Let $\mathcal{Z}$ be the largest part of $\mathcal{P}$\;
$\mathcal{Q}\leftarrow\mathcal{P}$;
$K\leftarrow \{\mathcal{Z}\}$;
$L\leftarrow \{\mathcal{X} \mid \mathcal{X}\neq\mathcal{Z},\mathcal{X}\in\mathcal{P}\}$\;
\lnl{ligne:mainwhile}	\While{$L\cup K\neq\emptyset$}{
	\lIf{there exists $\mathcal{X}\in L$}{
		$S\leftarrow\mathcal{X}$ and
		$L\leftarrow L\setminus \{\mathcal{X}\}$\;		
		}
	\Else{
\lnl{ligne:sommet}		Let $\mathcal{X}$ be the first part $K$ and $x$ arbitrarily selected in $\mathcal{X}$\;
		$S\leftarrow \{x\}$ and $K\leftarrow K\setminus \{\mathcal{X}\}$\;
		}	
	\ForEach{vertex $x\in S$}{
		\ForEach{part $\mathcal{Y}\neq\mathcal{X}$ such that $N(x)\overlap \mathcal{Y}$}{
			\lnl{ligne:insertion} Replace in $\mathcal{Q}$, $\mathcal{Y}$ by $\mathcal{Y}_1=\mathcal{Y}\cap 
			N(x)$ and $\mathcal{Y}_2=\mathcal{Y}\setminus N(x)$\;
			Let $\mathcal{Y}_{min}$ (resp. $\mathcal{Y}_{max}$) be the smallest part 
			  (resp. largest) among $\mathcal{Y}_1$ and $\mathcal{Y}_2$\;
			\lIf{$\mathcal{Y}\in L$}{
				$L\leftarrow
				L\cup\{\mathcal{Y}_{min},\mathcal{Y}_{max}\}\setminus\{\mathcal{Y}\}$\;
				}
			\Else{
				$L\leftarrow L\cup\{\mathcal{Y}_{min}\}$\;
				\lIf{$\mathcal{Y}\in K$}{
					Replace $\mathcal{Y}$ by $\mathcal{Y}_{max}$ in $K$\;
					}
				\lElse{Add $\mathcal{Y}_{max}$ at the end of $K$\;}
					
				}
			}
		}
	}
}
}
\end{algorithm2e}

To implement this rule, the parts are stored in two disjoint lists $K$ and $L$. The neighbourhoods of all the vertices of parts belonging to $L$ will be used to refine the partition. For the parts belonging to $K$, only the neighbourhood of one arbitrarily selected vertex is used. Since  $K$ is managed with a FIFO priority rule, this guarantees that the first part of the list, when extracted, is a module.

\begin{theorem} \label{th:partition-modulaire}
Let $\mathcal{P}$ be a partition of the vertices of a graph $G=(V,E)$. Algorithm~\ref{alg:part-modulaire} computes the coarsest modular partition for $G$ and $\mathcal{P}$ in time $O(n+m\log n)$.
\end{theorem}

The correctness of the algorithm follows from the next three invariant properties. The first invariant shows that a module contains in some part of the given partition cannot be split, while the third one guarantees that the algorithm outputs a modular partition.

\begin{enumerate}
\item \emph{If $M$ is a module of $G$ contained in a part $\mathcal{X}\in\mathcal{P}$, then there exists a part $\mathcal{Y}$ of the current partition containing $M$.} 
\item \emph{If $L=\emptyset$, then the first part $\mathcal{Y}$ of $K$ is a module.} 
\item \emph{If the current partition contains a part $\mathcal{X}$ that is not a module, then there exists $\mathcal{Y}\in L\cup K$ different from $\mathcal{X}$ and containing a splitter $y$ for $\mathcal{X}$.} %
\end{enumerate}
\textit{Complexity issues:} The main while loop (line~\ref{ligne:mainwhile}), manages a set $S$ of vertices whose neighbourhoods have to be used to refine the current partition. The set $S$ is computed from the lists $L$ and $K$. Since the current part containing a given vertex can be added to $L$, only if its size is smaller than half of the size of the former part containing $x$, the neighbourhood of each vertex $x$ is guaranteed to be visited at most  $\log(|V|)$ times by the algorithm. Furthermore, when a vertex $x$ of a part $\mathcal{X}$ extracted from $K$ is used, neither $x$ nor none of the vertices of $\mathcal{X}$ is used again. This yields to a $O(\sum_{x \in V}\log(|V|).|N(x)|)$ complexity, as claimed.

\subsection{Bibliographic notes}

As already mentioned, the use of partition refinement technique dates to 1971 for the deterministic automata minimization problem~\cite{Hop71}. In 1987, Paigue and Tarjan used again this technic to solve three different problems: functional partition, coarsest relational partition problems and doubly lexicographic ordering of a boolean matrix. In the late 90's, it has been used more systematically in the context of modular decomposition and transitive orientation 
yielding  $O(n+m\log n)$ practical and simple algorithms (see e.g. \cite{MS00,HMPV00}).

\section{Recursive computation of the modular decomposition tree}
\label{sec:algo}

In 1994, Ehrenfeucht, Gabow, McConnell and Sullivan~\cite{EGMS94} proposed a quadratic algorithm for the modular decomposition\footnote{This algorithm is designed for $2$-structures, a classical generalization of graphs.}. The principle of this algorithm, which we will call the \emph{skeleton algorithm}, is the basis of a large number of the known subquadratic algorithms proposed in the late 90's (see \eg~\cite{MS00,DGM01}), which could abusively be considered as a series of different implementations of the skeleton algorithm. The complexity of these implementations are respectively $O(n+m.\alpha(n,m))$ or $O(n+m)$~\cite{DGM01}, and finally $O(n+m\log n)$~\cite{MS00}. We describe the principle of the skeleton algorithm without considering the complexity issues. We then discuss the differences in the time complexity of the known algorithms.

\subsection{The skeleton algorithm}

Let us first mention that the skeleton algorithm computes a \emph{non-reduced} form of the modular decomposition tree $MD(G)$: the resulting tree may contain some series (or parallel) node child of a series (or parallel) node. All the algorithms we describe in this section will do so. It does not impact the complexity issues as a single search of the tree is enough reduce it in time $O(n)$. In the following, we will abusively denote $MD(G)$ the (non-reduced) decomposition tree returned by these algorithms.

The main idea developed by Ehrenfeucht et al.~\cite{EGMS94} is to first compute a \emph{"spine"} of the modular decomposition tree $MD(G)$, then to recursively compute the modular decomposition trees of some induced subgraphs which are eventually padded to the spine. More formally:

\begin{definition}
Let $v$ be an arbitrary vertex of a graph $G=(V,E)$. The \emph{$v$-modular partition} is the following modular partition:\\
\centerline{$\M(G,v)=\{v\}\cup\{M\mid M\mbox{ \textit{is a maximal module not containing} }v\}$}

We define $spine(G,v)$ as the modular decomposition tree $MD(G_{/\M(G,v)})$.
\end{definition}

First we notice that $\M(G,v)$ is easy to  compute.

\begin{lemma}\label{fallaitmettreunlemme}
The partition $\M(G,v)$ is the coarsest modular partition for $G$ and $\mathcal{P}=\{N(v), v, \overline{N}(v)\}$  and can be computed in time $O(n+m\log n)$.
\end{lemma}
\begin{figure}[htp]
\centerline{\scalebox{1}{\includegraphics[width=12cm]{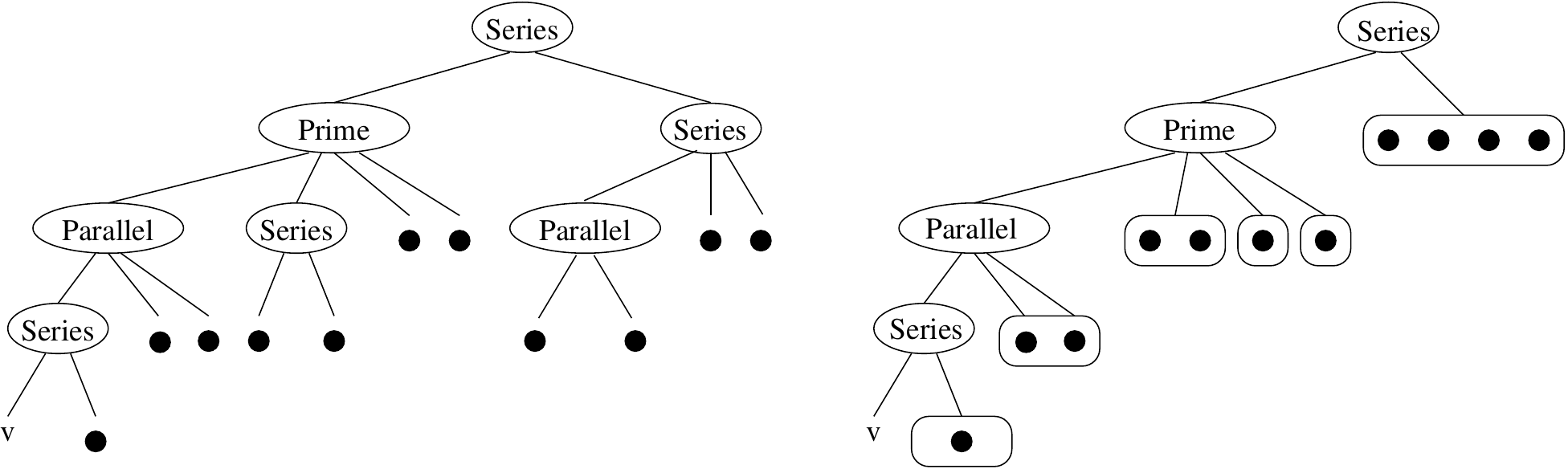}}}
\caption{
On the left, a modular decomposition tree $MD(G)$ and on the right, the modular partition $\M(G,v)$ with the corresponding spine  between $v$ and the root of $MD(G)$.
\label{fig:ehrenfreucht}}
\end{figure}

\begin{algorithm2e}[htb]
\small{
\caption{Ehrenfeucht et al.~\cite{EGMS94} \label{alg:ehrenfreucht}}
\KwIn{An arbitrary vertex $v$ of $G=(V,E)$, $T=spine(G,v)$ and $\{T_{X}=MD(G[X])\mid X\in \M(G,v)\}$}
\KwOut{The modular decomposition tree  $MD(G)$}
\Begin{
\ForEach{leaf $X$ of $T$}{
	Let $T_X=MD(G[X])$ and $p(X)$ be $X$'s father in $T$\;
	Replace $X$ by $T_X$ in $T$\;
	\lnl{ligne:clean} \If{the root $r(T_X)$ and $p(X)$ are both parallel or series}{
		Remove $r(T_X)$ and connect the children of  $r(T_X)$ to $p(X)$
		}
	}
}
}
\end{algorithm2e}

Let us notice that any degenerate strong module (series or parallel) containing $v$ will be represented in $spine(G,v)$ by a binary node. 
The purpose of test of Line~\ref{ligne:clean} in Algorithm~\ref{alg:ehrenfreucht} is to correctly fixed those binary nodes. The correctness of Algorithm~\ref{alg:ehrenfreucht} is a consequence of the following properties:

\begin{lemma} \label{lem:ehrenfreucht}\cite{EGMS94}
Let $v$ be a vertex of a graph  $G=(V,E)$ and $\M(G,v)$ be the associated modular partition. Then:
\begin{enumerate}
\item  Any non-trivial module of $G_{/\M(G,v)}$ contains $v$;
\item A set $\X\subset\M(G,v)$ is a non-trivial strong module of  $G_{/\M(G,v)}$ iff $\bigcup_{M\in\X} M$ is an ancestor of $v$ in $MD(G)$;
\item Any module not containing $v$ is a subset of a part  $M\in\M(G,v)$.
\end{enumerate}
\end{lemma}

Computing $spine(G,v)$ is a the difficult and technical task of the skeleton algorithm, indeed it is its main complexity bottleneck. The solution we present hereafter has been proposed in~\cite{EGMS94} and yields quadratic running time. Later on, Dahlhaus et al.~\cite{DGM01} improved this step and obtained a subquadratic running time (see discussion of Section~\ref{sub:quasilinear}).

\subsection{Computation of $spine(G,v)$.}

\begin{definition}
A graph $G=(V,E)$ is \emph{nested} if there exists a vertex $v\in V$ which is contained in all the non-trivial modules of $G$. Such a vertex is called an \emph{inner} vertex of $G$. 
\end{definition}

As a direct consequence of Lemma~\ref{lem:ehrenfreucht}, the quotient graph $G_{/\M(G,v)}$ is a nested graph with inner vertex $v$.

In order to compute the modules of $G_{/\M(G,v)}$ and $spine(G,v)$, Ehrenfeucht et al.~\cite{EGMS94} introduced an auxiliary \emph{forcing digraph} the arc set of which guarantees the existence of a directed path from any vertex $u$ to any vertex $w\in m(u,v)$, the smallest module containing $u$ and $v$. As $v$ belongs to all the modules of $G_{/\M(G,v)}$, a simple search on the forcing graph will suffice to compute $spine(G,v)$.

\begin{definition} \label{def:forcage}  \footnote{~The definition  proposed here slightly differs from the original one of~\cite{EGMS94}. This modification simplifies the relationships with the results of \cite{DGM01}.}
Let $v$ be an arbitrary vertex of a graph $G=(V,E)$. The \emph{forcing graph} $\F(G,v)$ is a directed graph whose vertex set is $V\setminus\{v\}$. The arc $\arc{xy}$ exists if $y$ is a splitter for  $\{x,v\}$.
\end{definition}

In other words, if $\arc{xy}$ exists then $y$ belongs to any module containing $v$ and $x$.

\begin{figure}[htp]
\centerline{\scalebox{1}{\includegraphics[width=10cm]{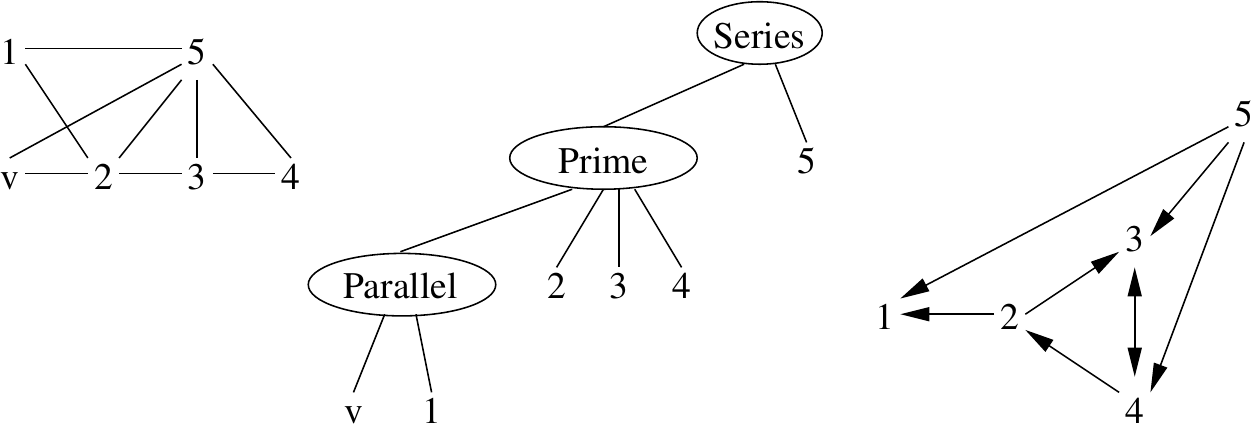}}}
\caption{A nested graph $G=(V,E)$ together with its modular decomposition tree $MD(G)$ and on its right the forcing graph $\F(G,v)$. The strongly connected components of $\F(G,v)$ are $\{1\},\{2,3,4\},\{5\}$. Any module of $G$ containing $3$ and $v$ also contains $\{1,2,4\}$, the vertices that can be reached from vertex $3$ in $\F(G,v)$.}\label{fig:forcage} 
\end{figure}

\begin{lemma} \cite{EGMS94} \label{lem:cfc}
If $X$ is the set of vertices that can be reached from vertex $x$ in the forcing graph $\F(G,v)$, then $\{v\}\cup X=m(v,x)$.
\end{lemma}

In the following we will only consider  the graph $G_{/\M(G,v)}$ and its forcing graph $\F(G_{/\M(G,v)}, v)$.  Applying Lemma~\ref{lem:cfc} to $\F(G_{/\M(G,v)}, v)$, we obtain the following property.

\begin{corollary} \cite{EGMS94} \label{cor:cfc}
Let $M_x$ be the module of $\M(G,v)$ containing the vertex $x$. If $\X$ is the set of modules that can be reached from $M_x$ in $\F(G_{/\M(G,v)}, v)$, then $\bigcup_{M\in\X} M=m(v,x)$.
\end{corollary}

We now consider the \emph{block graph} $\B(G,v)$ of $\F(G_{/\M(G,v)}, v)$ (see~\cite{CLR90}) whose vertices are the strongly connected components of
 $\F(G_{/\M(G,v)}, v)$, also called the \emph{blocks} of $(G,v)$. An arc of $\B(G,v)$ between the block $B$ and $B'$ exists if the vertices of $B'$ can be reached in $\F(G_{/\M(G,v)}, v)$  from the vertices of $B$.

\begin{lemma} \cite{EGMS94} \label{lem:chemin}
The transitive reduction of the block graph $\B(G,v)$ is a chain.
\end{lemma}

A set of vertices of a digraph is a \emph{sink} if it has no out-neighbour. By Lemma~\ref{lem:chemin}, any sink set of $\F(\M(G,v))$  is the union of consecutive blocks containing the last one in the transitive reduction of $\B(G,v)$. Each sink set corresponds to a module of $G_{/\M(G,v)}$.

\begin{corollary} \cite{EGMS94}
Let $v$ be a vertex of a graph $G=(V,E)$. A set $M$ of vertices containing $v$ is a module of $G_{/\M(G,v)}$ iff $M$ is the union of $\{v\}$ and the modules of $\M(G,v)$ belonging to a sink set $X$ of $\B(G,v)$.
\end{corollary}

Thereby the forcing graph $\F(G_{/\M(G,v)}, v)$ describes the modules of $G_{/\M(G,v)}$ and the block graph $\B(G,v)$ allows us to compute $spine(G,v)$. Finally, $MD(G)$ is obtained recursively by following the lines of Lemma~\ref{lem:ehrenfreucht}.

\subsection{Complexity issues}
\label{sub:quasilinear}

Rather than detailing the complexity analysis, we point out the differences between the original skeleton algorithm presented in~\cite{EGMS94} and its later versions improved in~\cite{DGM01}. The interested reader should access the original papers for details. As already mentioned, a quadratic time complexity analysis is proposed in~\cite{EGMS94}. The main bottlenecks are the computation of the partition $\M(G,v)$ and the construction of $MD(G_{/\M(G,v)})$.

Two new versions of the skeleton algorithm proposed by Dahlhaus, Gustedt and McConnell~\cite{DGM01}, respectively run in $O(n+m.\alpha(n,m))$ time and in linear time. To improve the time complexity, the authors of~\cite{DGM01} borrowed from \cite{Dah95} the idea to first recursively compute the modular decomposition trees of the subgraphs induced by $N(v)$ and by $\overline{N}(v)$. It follows from the next Lemma, that $\M(G,v)$ is easy to retrieve from those trees.

\begin{lemma} \label{lem:submodule1}
If $\mathcal{X}$  is a module of $\M(G,v)$, then $\mathcal{X}$ is either a module of $G[N(v)]$ or a module of $G[\overline{N}(v)]$.
\end{lemma}

As in~\cite{EGMS94}, the technique used to compute $spine(G,v)$ relies on a forcing digraph. Remind that the vertices of $\F(\M(G,v))$ are the modules of $G$ (indeed the modules of $\M(G,v)$) which turns out to be a too strong condition for time complexity issues. In~\cite{DGM01}, the forcing digraph is rather defined with the help of an equivalence relation. The idea is that each equivalence class gathers vertices of $N(v)$ or of $\overline{N}(v)$ which appear in a set of sibling modules of some ancestor node of $v$ in $MD(G)$ (or $spine(G,v)$). The partition defined by the equivalence classes is a coarser partition than $\M(G,v)$.

The final trick is that given $MD(G[N(v)])$ and $MD(G[\overline{N}(v)])$, the computation of $\M(G,v)$, $spine(G,v)$ and finally $MD(G)$ has to be done in time linear in the number of \emph{active edges}, \ie the edges incident to $v$ and the edges linking vertices of $N(v)$ and $\overline{N}(v)$. The $\alpha(n,m)$ factor in the first version of the skeleton algorithm presented in~\cite{DGM01} is due to the use of some \emph{union-find} data-structures required to update the current tree. A clever time complexity analysis yields linear time if a careful pre-processing step is used to fix the recursion tree. 

\subsection{Bibliographic notes}
\label{sec:biblio-refine}

Let us mention that  the problem of finding a simple linear time algorithm for the modular decomposition is presented in~\cite{MS00} or \cite{Spi03} as an open problem. In its book~\cite{Spi03}, Spinrad wrote p.149:
 \begin{quote}
 "\emph{I hope and believe that in a number of years the linear algorithm can be simplified as well"}
 \end{quote}

Based on partition refinement techniques, a simplified $O(n+m\log n)$ version of the skeleton algorithm has been developed in~\cite{MS00}. 

\section{Factoring permutation algorithm}

In its PhD Thesis, Capelle~\cite{Cap97} proved that computing the modular decomposition tree of a graph and computing a \emph{factoring permutation} (see Definition~\ref{def:perm-fact} and Figures~\ref{fig:md-tree},~\ref{fig:fract-tree}) 
are two equivalent tasks, as one can be retrieved from each another in linear time~\cite{CHM02}. It follows that computing the modular decomposition of a graph can be divided into two different steps: 1) computation of a factoring permutation; 2) computation of the modular decomposition tree given the factoring permutation. The main interest of such a strategy is to obtain an algorithm that avoids the auxiliary data-structures needed to compute  \emph{union-find} and \emph{least common ancestor} operations, as used in~\cite{DGM01} for example. Moreover, in some recent applications (\emph{e.g.} comparative genomics~\cite{UY00,BHS02,HMS09}), the given data is not the graph nor the partitive family but rather a factoring permutation. This concept turns out to be of interest by itself.

As noticed by Capelle~\cite{Cap97}, this strategy was already used in few cases such as the computation of the modular decomposition tree of chordal graph~\cite{HM91} and the block tree of inheritance graphs~\cite{HHS95}. In~\cite{HPV98,HPV99}, a partition refinement algorithm is proposed to compute a factoring permutation of a graph in time $O(n+m\log n)$. Restricted to cographs, the complexity can be improved down to linear time~\cite{HP05}.

\medskip
We will first revisit Algorithm~\ref{alg:affinage} of~\cite{HPV98} and show how it can be adapted to compute a factoring permutation in time $O(n+m\log n)$. This algorithm has to be compared to the McConnell and Spinrad's implementation~\cite{MS00} of Ehrenfeucht et al.'s algorithm. The main differences are that the modular decomposition tree is never built and the relative order between the different parts of the partition is important.

There exist several linear time algorithms that given a factoring permutation of a graph compute its modular decomposition tree. A recent one is proposed in~\cite{BCMR05,BCMR08}. We describe the principle of the first one due to Capelle, Habib and de Montgolfier~\cite{CHM02}.

\subsection{Computing a factoring permutation}

An ordered partition $\P=[\X_1,\dots,\X_k]$ of a set $\mathcal{E}$ defines a partial order on $\mathcal{E}$, the maximal antichains  of which are exactly the parts of $\P$. In other words, we have $x_i<_{\P}x_j$ iff $x_i\in \X_i$, $x_j\in\X_j$ and $i<j$. Thereby refining an ordered partition could be understood as computing an extension of the corresponding partial order. 

We will abusively write $x<_{\P} M$, for $x\in \mathcal{E}$ and $M\subset \mathcal{E}$, if $x<_{\P} y$ for all $y\in M$. 
To prove the correctness of the algorithm, we need to generalize the definition of interval of permutations to ordered partitions.

\begin{definition}
Let $\P$ be an ordered partition of a set $\mathcal{E}$. A subset $S\subseteq\mathcal{E}$ is an \emph{interval} of $\P$ iff there are two parts $\mathcal{L}\in\P$ and $\mathcal{R}\in\P$ (not necessarily distinct) intersecting $S$ 
such that for any part $\X$:
\vspace{-0.4cm}
\begin{itemize}
\item if $\mathcal{L}<_{\P} \X<_{\P}\mathcal{R}$, then $\X\subset S$;
\item if $\X<_{\P}\mathcal{L}$ or $\mathcal{R}<_{\P}\X$, then $\X\cap S=\emptyset$.
\end{itemize}
\end{definition}

\medskip
To compute a factoring permutation, the main steps of the algorithm we present are: 1) computation of an ordered partition that is a modular partition $\M(G,v)$ such that the strong modules containing a vertex $v$ are intervals of $\M(G,v)$; and 2) recursive computation of a factoring permutation of each of the subgraphs induced by a module $M\in\M(G,v)$.

\begin{algorithm2e}[h]
\caption{\emph{Factoring-permutation($G$, $v$)} \label{alg:permutation-fact}}
\KwIn{A graph $G=(V,E)$ and a vertex $v\in V$}
\KwOut{A factoring permutation of $G$}
\Begin{
Let $\mathcal{P}=[\overline{N}(v),\{v\},N(v)]$ be an ordered partition\;
Apply Algorithm~\ref{alg:part-modulaire} with the following refinement rule\;
Let $x$ be the current pivot vertex and $\mathcal{Y}$ a part such that $N(x)\overlap\mathcal{Y}$\;

\eIf{$x\leqslant_{\P} v\leqslant_{\P} \mathcal{Y}$ or $\mathcal{Y}\leqslant_{\P} v\leqslant_{\P} x$}{
	Substitute $\mathcal{Y}$ by $[\mathcal{Y}\cap\overline{N}(x),\mathcal{Y}\cap N(x)]$\;
	}{
	Substitute $\mathcal{Y}$ by $[\mathcal{Y}\cap N(x),\mathcal{Y}\cap\overline{N}(x)]$\;
	}	

\ForEach{part $\mathcal{X}\in\M(G,v)$, such that $|\X|>1$}{
	Let $x$ be the last vertex of $\X$ used as pivot\;
	$\P_{\X}\leftarrow$ Factoring-permutation($G[\X]$, $x$)\;
	Substitute $\X$ by $\P_{\X}$\;
	}
}
\end{algorithm2e}

\begin{figure}[htp]
\centerline{\scalebox{1}{\includegraphics[width=13cm]{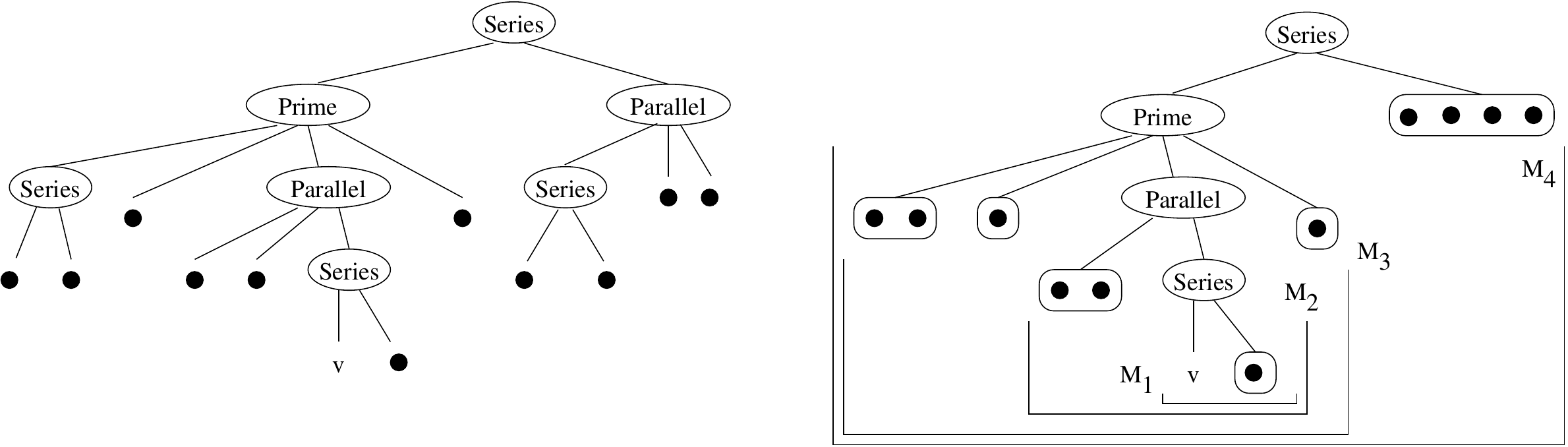}}}
\caption{Layout of the modular decomposition tree $MD(G)$ such that the neighbours of $v$ are placed on the right of $v$ and the non-neighbours on the left. The right tree enlights the modules of $\M(G,v)$ and the strong modules $M_1,M_2,M_3$ and $M_4$ containing $v$. Algorithm~\ref{alg:permutation-fact} first computes the partition $\M(G,v)$ and then recursively solves the problem on each module of $\M(G,v)$} 
\label{fig:permutation-fact}
\end{figure}

\begin{theorem} \label{th:logn}
Algorithm~\ref{alg:permutation-fact} compute in time $O(n+m\log n)$ a factoring permutation of a graph $G=(V,E)$.
\end{theorem}
\begin{proof}
Using lemma \ref{fallaitmettreunlemme} $\M(G,v)$ can be computed in $O(n+m\log n)$. By Lemma~\ref{lem:ehrenfreucht}, any module not containing $v$ is a subset of some module of $\M(G,v)$. It thereby suffices to prove that the following invariant is satisfied by Algorithm~\ref{alg:permutation-fact} (see Figure~\ref{fig:permutation-fact}):
 $$\Pi=\mbox{\emph{any strong module containing $v$ is an interval of the current partition}}$$
The property $\Pi$ is obviously satisfied by  the initial partition $[\overline{N}(v),\{v\},N(v)]$. Assume by induction $\Pi$ holds before the current partition $\mathcal{P}$ is refined by $N(x)$ for some vertex $x$. Let $M$ be a module containing $v$ and $\X$ be a part of $\P$ such that $\X\overlap N(x)$. There are two distinct cases:
\begin{itemize}
\item $x\notin M$: no vertex $y$ of $\X\cap N(x)$ belongs to $M$, otherwise $x$ would be a splitter for $v$ and $y$; 
\item $x\in M$:  if $\X\subset N(v)$, then any vertex $y\in\X\cap \overline{N}(x)$ belong to $M$, otherwise $y$ would be a splitter for $x$ and $v$. Similarly if  $\X\subset \overline{N}(v)$, then any vertex $y\in\X\cap N(x)$ belongs to $M$.
\end{itemize}
It follows that $\P'=$\emph{Refine}$(\P,N(x))$ also satisfies the invariant $\Pi$. The complexity analysis is similar to the analysis of Algorithm~\ref{alg:part-modulaire}.
\end{proof}

\subsection{The case of cographs}

The natural question is how to get rid of the $\log n$ factor in the complexity of Algorithm~\ref{alg:permutation-fact}. Restricting the problem to cographs (or totally decomposable graphs - see Section~\ref{sub:cographs}) gives some ideas. The reader should keep in mind that the $\log n$ factor corresponds to the number of times the neighbourhood of a vertex can be used to refine the partition. So, a linear time algorithm should use each vertex as a pivot a constant number of times.

The linear time cograph recognition algorithm proposed in~\cite{HP05} computes a factoring permutation as a preliminary step. It roughly proceeds as follows. It uses at most one vertex per partition part to refine the ordered partition $[\overline{N}(v),\{v\},N(v)]$. Assuming the input graph is a cograph, when none of the parts of the current partition is free of pivot, it can be proved that one of the two non-singleton parts closest to $v$ in the current partition, say $\X$, can be refined into $[\overline{N}(x)\cap \X,\{x\},N(x)\cap \X]$ ($x$ being the used pivot of $\X$). This step creates at least one new part free of pivot and thereby relaunches the refining process.

\subsection{From factoring permutation to modular decomposition tree}
\label{sec:perm-to-tree}


As already noticed, a natural idea to compute the modular decomposition tree is to compute for each pair  $x,y$ of vertices the set of \emph{splitter} $S(x,y)$. Unfortunately a linear time algorithm could not afford the computation of all these $O(n^2)$ sets. But if one has in hand a factoring permutation $\sigma$, it is then sufficient to consider the pairs of consecutive vertices in $\sigma$. Indeed, Capelle et al.'s algorithm~\cite{CHM02} only computes for each pair of vertices $x=\sigma(i)$ and $y=\sigma(i+1)$ ($i\in[1,n-1]$) the leftmost and the rightmost (in $\sigma$) splitter of $x$ and $y$. These two splitters define two intervals of $\sigma$, which are both contained in $m(x,y)$, the smallest module containing both $x$ and $y$:

\begin{itemize}
\item the \emph{left fracture} $F_{l}(x,y)=[z,x]$ if $z$ is the leftmost splitter of $\{x,y\}$ in $[\sigma(1),y]$ (if any);
\item the \emph{right fracture} $F_{d}(x,y)=[y,z]$ if $z$ is the rightmost splitter of $\{x,y\}$ in $[x,\sigma(n)]$ (if any).
\end{itemize}

\begin{figure}[htp]
\centerline{\scalebox{1}{\includegraphics[width=11cm]{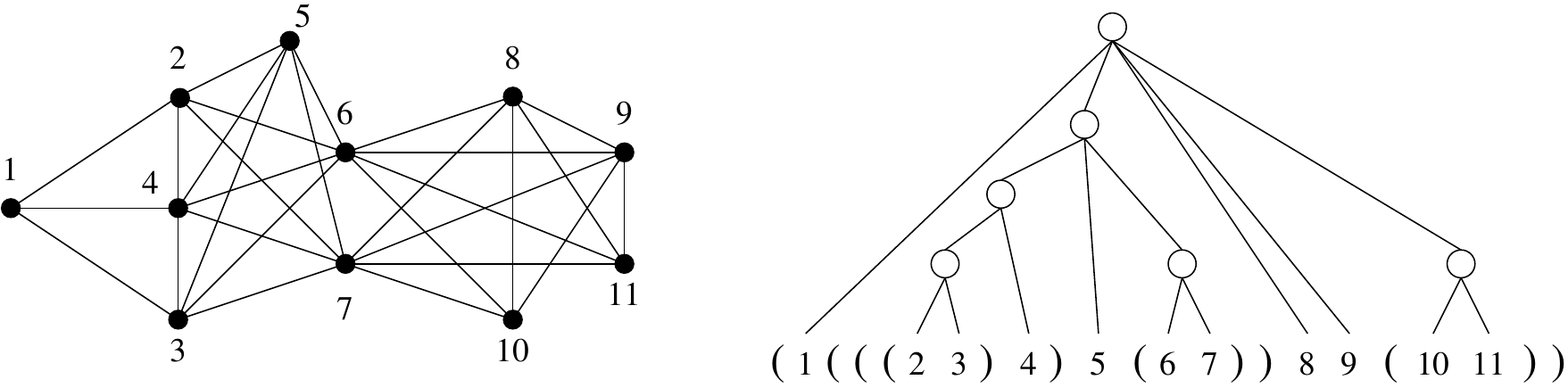}}}
\caption{A graph $G=(V,E)$ for which $\sigma=\textbf{1~2~3~4~5~6~7~8~9~10~11}$ is a factoring permutation (see Definition~\ref{def:perm-fact}). The right fracture of $(3,4)$ does not exist but $Fg(3,4)=[2,3]$. We also have $Fd(1,2)=[2,7]=Fg(7,8)$.}
\label{fig:fract-tree}
\end{figure}

The set of fractures (left and right) defines a parenthesis system. Forgetting the initial pairing of the parenthesis, this system
naturally yields a tree, called the \emph{fracture tree} and denoted $FT(G)$ (see Figure~\ref{fig:fract-tree}).  The fracture tree is actually a good estimation of the $MD(G)$ (see Lemma~\ref{lem:fracture}) which can be computed in linear time by two traversals of $\sigma$: the first traversal computes the fractures, the second builds the tree. 

\begin{lemma} \cite{CHM02} \label{lem:fracture}
Let $\sigma$ be a factoring permutation of a graph $G$ and $M$ be a strong module of $G$. If $M$ is a prime node of $MD(G)$ and if the father of $M$ is a degenerate, then there exists a node $N$ of the fracture tree $FT(G)$ such that $M$ is the set of leave of the subtree of $FT(G)$ rooted at $N$
\end{lemma}

For example, in Figure~\ref{fig:fract-tree}, any strong module but $M=\{8,9,10,11\}$ is represented by some node of $FT(G)$. Let us notice that the above lemma does not implies that the strong module  $\{2,3,4\}$ has a corresponding node in $FT(G)$.

Henceforth to compute $MD(G)$, the fracture tree $FT(G)$ has to be cleaned. To that aim, Capelle et al.~\cite{CHM02} use four extra traversals of the factoring permutation. The first one identifies the strong modules represented by some nodes of $FT(G)$; the second finds the dummy nodes of $FT(G)$; the third search for strong modules that are merged in a single node of $FT(G)$; and the last one remove the nodes of $FT(G)$ that does not represesent strong modules. The complexity of each of these four traversals is linear in the size of $G$, $O(n+m)$.

\subsection{Bibliographic notes}

An attempt to generalize to arbitrary graphs  the linear time algorithm which computes a factoring permutation of a cograph has been proposed in~\cite{HMP04}. Unfortunately the algorithm of~\cite{HMP04} contains a flaw. The recent linear time modular decomposition algorithm presented in~\cite{TCHP08} mixes the ideas from the factoring permutation algorithms and the skeleton algorithm. It generalizes the ordered partition refining technique to tree partition and avoids union-find or least-common ancestor data-structures. In that sense this new algorithm may be considered as a positive answer to Spinrad's comment (see Section~\ref{sec:biblio-refine}).

\section{Three novel applications of the modular decomposition}

As mentioned in the introduction modular decomposition is used in a number of algorithmic graph theory applications and more generally applies to various discrete structures (see~\cite{MR84}). We conclude this survey with the presentation of three novel applications which are good witnesses of the use of modular decomposition. The first one is a pattern matching problem which is closely related to the concept of factoring permutations. The second one provides an example of dynamic programming on the modular decomposition tree in the context of comparative genomic. Finally, we list a series of parameterized problems for which module based data-reduction rules leads to polynomial size kernels.

\subsection{Pattern matching - common intervals of two permutations}

Motivated by a series of genetic algorithms for sequencing problems, \emph{e.g.} the TSP, Uno and Yagiura~\cite{UY00} formalized the concept of \emph{common interval} of two permutations. As we will see in the next subsection, in the context of comparative genonic, common intervals reveal conserved structures in chromosomal material.

\begin{definition}
A set $S$ of elements is a \emph{common interval} of a set of permutations $\Sigma$ if in each permutation $\sigma\in\Sigma$, the elements of $S$ form an interval of $\sigma$ (see Section~\ref{sec:permutation} for the definition of an interval). 
\end{definition}

It is fairly easy to observe that the family $\mathcal{I}$ of common intervals of two permutations is a weakly partitive family (see Definition~\ref{def:partitive}) and thus all the results from the theory presented in Section~\ref{sec:partitive} apply. In particular, the set of strong common intervals are organized into a tree, namely the \emph{strong interval tree}.

\begin{figure}[htp]
\centerline{\scalebox{1}{\includegraphics[width=8cm]{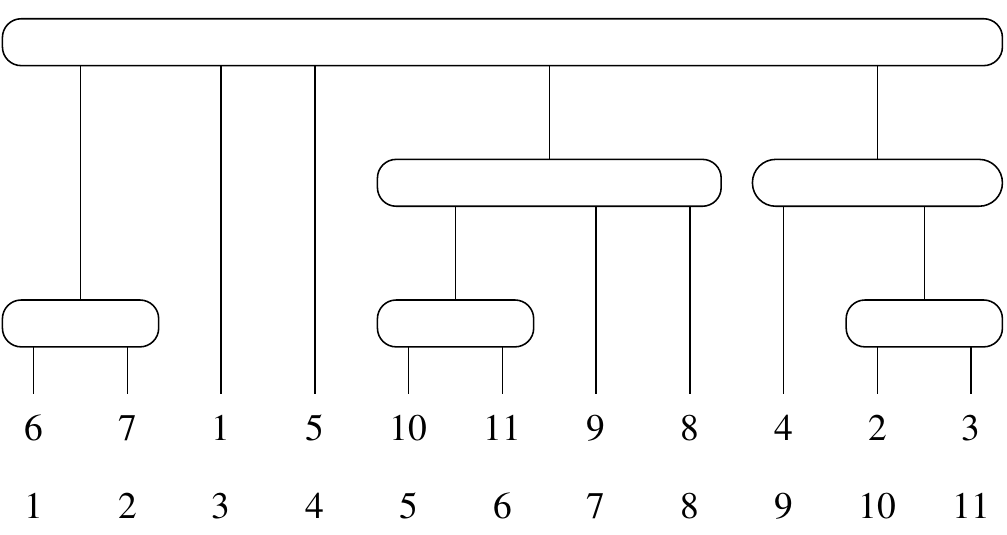}}}
\caption{The strong interval tree of two permutations. Remark $\{9,10,11\}$ and $\{8,9\}$ are also a common interval, but they are not strong as they overlap.}
\label{fig:interval-commun}
\end{figure}

Despite of the existence of the (weakly) partitive set theory for more that thirty years, the natural concept of interval substitution and decomposition appeared only very recently in the context of the combinatorial study of permutations (see \eg ~\cite{AS02,AA05}). Atkinson and Stitt~\cite{AS02} (re)discovered the concept of substitution under the name of \emph{wreath product}. In 2005, Albert and Atkinson showed that, if the number of \emph{simple} (\ie prime) permutations in a pattern restricted class of permutations is finite, the class
has an algebraic generating function and is defined by a finite set of restrictions. More recently, Bouvel, Rossin and Viallette~\cite{BR06,BRV07} used the strong interval tree to solve the longest common pattern problem between two permutations.

Uno and Yagiura~\cite{UY00} proposed the first linear time algorithm to enumerate the common intervals of two permutations. More precisely, it runs in $O(n+K)$ time, where $K$ is the number of those common intervals (which is possibly quadratic). Alternative algorithms have been recently proposed~\cite{HMS09,BCMR08}. We sketch Uno and Yagiura's algorithm and discuss how it can be genralized to compute the modules of a graph when a factoring permutation is given.

Without loss of generality, we will consider the problem of computing the common intervals of a permutation $\sigma$ and the identity permutation $\mathbb{I}_n$. To identify the common intervals of a permutation $\sigma$ and $\mathbb{I}_n$, the algorithm traverses $\sigma$ only once. We denote by  $[i,j]$ the interval of $\sigma$ composed by the elements whose indexes are between $i$ and $j$ in $\sigma$: i.e.  $[i,j]=\{x\mid i\leqslant \sigma(x)\leqslant j\}$. An element $x\notin[i,j]$ is a \emph{splitter} of the interval $[i,j]$ if there exist $y\in [i,j]$ and $z\in [i,j]$ such that $y<x<z$. By $s([i,j])$ we denote the
number  of \emph{splitters} of the interval $[i,j]$.
The algorithm uses a list  \texttt{Potentiel} to filter and extract $\sigma$ the common intervals of $\sigma$ and $\mathbb{I}_n$. An element $i$ belongs to the list \texttt{Potentiel} as long as it may be the right boundary of a common interval.
The step $i$ consists in removing those elements which we know they cannot be the left boundary of a common containing. This filtering can be done efficiently by computing $s([i,j])$ (see Lemmas~\ref{lem:UY1} and~\ref{lem:UY2}).

\begin{algorithm2e}[h]
\label{alg:UY00}
\caption{Uno and Yagiura's algorithm~\cite{UY00}}
\KwIn{A permutation $\sigma$}
\KwOut{The set of intervals common to $\sigma$ and the identity permutation $\mathbb{I}_n$}
\Begin{
Let \texttt{Potentiel} be an empty list\;
\For{$i=n$ downto $1$}{
	(Filter) Remove from \texttt{Potentiel} the boundaries $r$ s.t. $\forall j\leqslant i$, $[j,r]$ is not a \textbf{common interval} of $\sigma$ and $\mathbb{I}_n$ \;
	(Extraction) Search \texttt{Potentiel} to find the boundaries $r$ s.t. $[i,r]$ is a \textbf{common interval} of $\sigma$ and $\mathbb{I}_n$ and output those intervals $[i,r]$\;
	(Addition) Add $i$ to \texttt{Potentiel}\;
	}
}
\end{algorithm2e}

The following properties are fundamental in the correctness of the algorithm:

\begin{lemma} \cite{UY00} \label{lem:UY1}
An interval $[i,j]$ of $\sigma$ is a common interval of $\sigma$ and $\mathbb{I}_n$ iff $s([i,j])=0$.
\end{lemma}

\begin{lemma} \cite{UY00,BHP05} \label{lem:UY2}
If $s([i,j])>s([i,j+1])$, then it does not exist $r<i$ such that $[r,j]$ is a common interval of $\sigma$ and $\mathbb{I}_n$.
\end{lemma}

The second lemma above means that if $s([i,j])>s([i,j+1])$ then the vertex $\sigma^{-1}(j+1)$ is a splitter of $[i,j]$. Thereby any common interval containing $[i,j]$ as a subset has to extend up to $\sigma^{-1}(j+1)$.

\paragraph*{Application to factoring permutations of a graph.} The most striking link between common intervals and modules of graphs is observed on permutation graphs (see Lemma~\ref{lem:interval-permutation}). \emph{Permutation graphs} are defined as the intersection graphs of a set of segments between two parallel lines (see~\cite{Gol80,BLS99} for example). It follows that the vertices of a permutation graph $G=(V,E)$ can be numbered from $1$ to $n$ such that there exists a permutation $\sigma$ of $[1,n]$ such that vertex numbered $i$ is adjacent to vertex numbered $j$ iff $i<j$ and $\sigma(j)<\sigma(i)$. The permutations $\sigma$ and $\mathbb{I}_n$ form the \emph{realizer} of $G$. As first observed by de Montgolfier, any permutation belonging to a realizer of a permutation graph is a factorizing permutation of that graph. It follows from Lemma~\ref{lem:fp-perm-partitive} that:

\begin{lemma}~\cite{Mon03} \label{lem:interval-permutation}
Let $G=(V,E)$ be a permutation graph and $(\mathbb{I}_n,\sigma)$ be its realizer. A set of vertices $M$ is a strong module iff $M$ is a strong common interval of $\mathbb{I}_n$ and $\sigma$.
\end{lemma}

The permutation graph corresponding to the permutations depicted in Figure~\ref{fig:interval-commun} is the graph $G$ of Figure~\ref{fig:md-tree}. Notice that the strong interval tree of these two permutations is isomorphic to the modular decomposition tree of $G$.

It follows from Lemma~\ref{lem:interval-permutation} that applied to the realizer of a permutation graph, Algorithm~\ref{alg:UY00} computes its strong modules. Though some extra work is required to obtained the modular decomposition tree, the complexity remains linear time. Moreover, as shown in~\cite{BHP05}, Uno and Yagiura's algorithm can directly be adapted to compute, given a factoring permutation, the strong modules of a graph. The number $s([i,j])$ becomes the number of splitters (in the sense of the modular decomposition, see Section~\ref{sec:module}) of the vertices contained in the interval $[i,j]$ of the factoring permutation. Now notice that Algorithm~\ref{alg:UY00} does not only output the strong common intervals. In order to restrict the enumeration to strong modules, a slight modification is required. A first traversal computes the strong right modules (i.e. the modules that are intervals of $\sigma$ and which are not overlapped on their right boundary by any other module). Then a second traversal can detect those modules which are overlapped on the left boundary.

\subsection{Comparative genomic - perfect sorting by reversals}

A \emph{reversal} in a permutation $\sigma$ consists in reversing the order of the elements of an interval of $\sigma$. When dealing with \emph{signed} permutations (whose elements are positive or negative), a reversal also flips the sign of the element of the reserved interval. Given two (signed) permutations $\sigma$ and $\tau$, the problem of \emph{sorting by reversals} asks for a series of reversals (a \emph{scenario}) to transform $\sigma$ into $\tau$.

Sorting by reversals is used in comparative genomic to measure the evolutionary distance between the genomes of two chromosomes, modeled as signed permutations~\cite{BHS02}. When comparing two genomic sequences, it can be assumed that the intervals having the same gene content are likely to have been present in their common ancestor and may witness to some functionally interacting proteins. Such a conserved genomic structure in the signed permutation model corresponds to common intervals. So to guess an evolutionary scenario between two genomic sequences represented by signed permutations $\sigma$ and $\tau$, one could asks for the smallest \emph{perfect scenario}, which is a series of reversals that preserves any common interval of $\sigma$ and $\tau$. For further details on this topic, the reader could refer to~\cite{BHS02,BBCP05}.

\begin{figure}[htp]
\centerline{\scalebox{1}{\includegraphics[width=8cm]{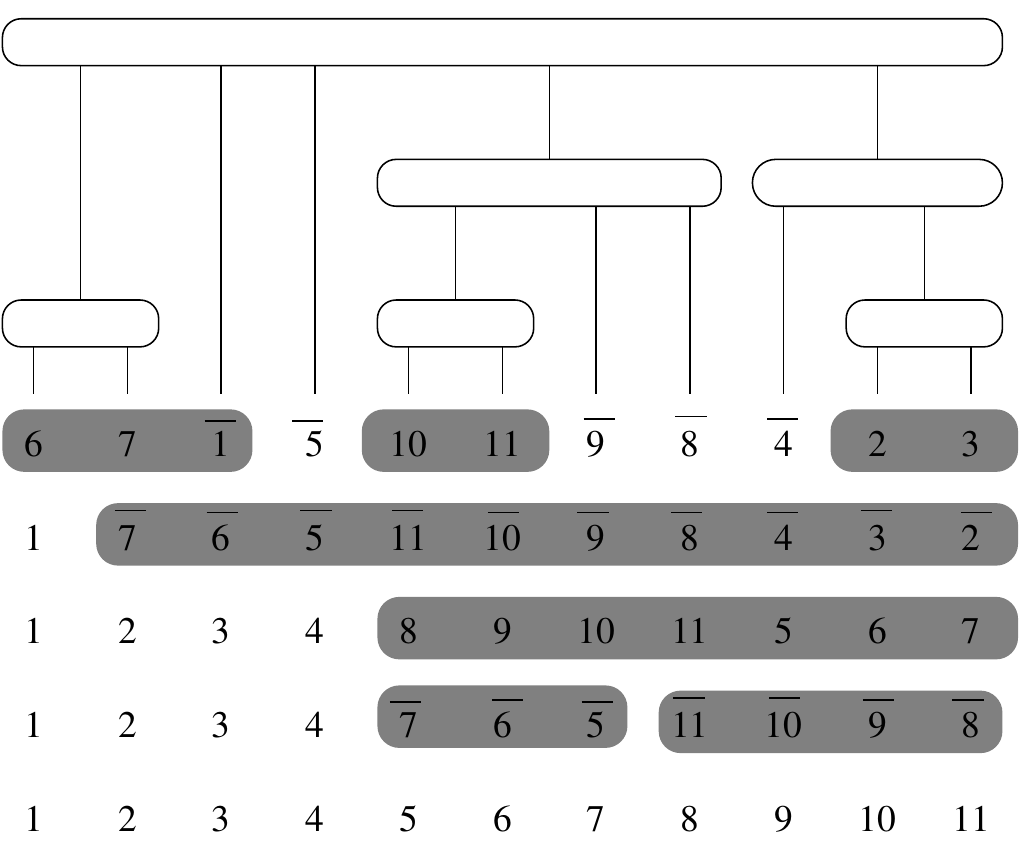}}}
\caption{A perfect scenario of length $7$.}
\label{fig:interval-commun-sc}
\end{figure}

As mentioned in the previous subsection, the set of common intervals of two permutations (signed or not) defines a weakly partitive family. It follows that one can distinguish prime from degenerate strong common intervals. As shown by the following lemma, we can read on the strong interval tree which are the perfect scenarios.

\begin{lemma}~\cite{BBCP05}
A reversal scenario for two signed permutations $\sigma$ and $\tau$ is perfect iff any reversed interval is either a prime common interval of $\sigma$ and $\tau$, or the union of strong common intervals which form a subset of the children of a prime common interval.
\end{lemma}

It follows from the previous lemma that the strong interval tree is useful to compute minimum perfect scenarios.  Indeed with some extra technical properties to deal with the signs it can be shown that a simple dynamic programming algorithm on the strong interval tree solves the problem in time $O(2^k\times n\sqrt{n\log n})$, where $k$ is the maximum number of prime nodes which are children of the same prime node. In practice, the parameter $k$ keeps very small~\cite{BCP08}: e.g. when comparing the chromosome $X$ of the mouse and the rat, we have $k=0$~\cite{BBCP07}.

\subsection{Parameterized complexity and kernel reductions - cluster editing}

The design of parameterized algorithms is, among others, one of the modern techniques to cope with NP-hard problems. A problem $\Pi$ is \emph{fixed parameter tractable} (FPT) with respect to parameter $k$ if it can be solved in time $f(k).n^{O(1)}$ where $n$ is the input size. The idea behind parameterized algorithms is to find a parameter $k$, as small as possible, which controls the combinatorial explosion.  Many algorithm techniques have been developed in the context of fixed parameter complexity, among which \emph{kernelization}. A parameterized problem $(\Pi,k)$ admits a \emph{polynomial kernel} if there is a polynomial time algorithm (a set of \emph{reduction rules}) that reduces the input instance to an instance whose size is bounded by a polynomial $p(k)$ depending only in $k$, while preserving the output. The classical example of parameterized problem having a polynomial kernel is the problem \textsc{vertex cover} parameterized by $k$ the solution size, which has a $2k$ vertex kernel. For textbooks on this topics, the reader should refer to~\cite{DF99,Nie06,FG06}.

Recently, the modular decomposition appeared in kernalization algorithms for a series of parameterized problems among which: \textsc{cluster editing}~\cite{Nie06}, \textsc{bicluster editing}~\cite{PSS07}, \textsc{fast} (feedback arc set in tournament)~\cite{DGH06b}, \textsc{closest 3-leaf power}~\cite{BPP09}, \textsc{flip consensus tree}~\cite{BBT08}. We discuss the \textsc{cluster editing} problem. Concerning the others, the reader should refer to the original papers.

The parameterized \textsc{cluster editing} problem asks whether the edge set of an input graph $G$ can be modified by at most $k$ modifications (deletions or insertions) such that the resulting graph $H$ is a disjoint union of cliques (e.g. clusters). This problem is NP-complete but can be solved in time $O^*(3^k)$ by a simple bounded search tree algorithm~\cite{Cai96}, which iteratively branches on at most $k$ $P_3$'s. Recent papers~\cite{Guo07,FLR07} showed the existence of a linear kernel (best bound is $4k$). The reduction rules used for these linear kernels are crown rules involving modules. For the sake of simplicity we only present the two basic reduction rules which leads to a quadratic vertex kernel.

\begin{lemma}
Let $G=(V,E)$ be a graph. A quadratic vertex kernel for the \textsc{cluster editing} problem is obtained by the following reduction rules:
\begin{enumerate}
\item Remove from $G$ the connected components which are cliques.
\item If $G$ contains a clique module $C$ of size at least $k+1$, then remove from $|C|-k-1$ vertex from $C$.
\end{enumerate}
\end{lemma}

It is clear that these rules can be applied in linear time using modular decomposition algorithms. The proof idea works as follows. The first rule is obviously safe. Concerning the second rule, simply observe that to disconnect a clique module of size $k+1$ from the rest of the graph, at least $k+1$ edge deletions are required. Now assuming $G$ is a positive instance, each cluster of the resulting graph $H$ can be bipartitioned into the vertices non-incident to a modified edge and the other vertices (the \emph{affected vertices}). Finally $k$ edge modifications can create at most $2k$ clusters and the total number of affected vertices is bounded by $2k$. This shows that the number of vertices in the reduced graph $H$ is at most $2k^2+4k$.

The \textsc{bicluster editing} problem edits the edge set of a graph to obtain a disjoint union of complete bipartite graphs. Instead of considering clique modules, we need to consider independent set modules~\cite{PSS07}. The proof is then slightly more complicated and relies on a careful analysis of the modification of the modular decomposition under edge insertion or deletion. In the case of \textsc{fast}, similar rules involving transitive modules also yields a quadratic kernel bound. Note that for these two problems, linear kernels can be obtained with more sophisticated reduction rules~\cite{GHKZ08,BFG09}.

\section{Conclusions and perspectives}

An important remaining open problem is the proposal of a simple linear time certifying algorithm for modular decomposition. In fact the algorithms described here produce a labelled tree that can be checked in linear time if they are decomposition trees. But for certifying  that some decomposition tree is the modular decomposition one must certify all node labels. The bottleneck is the certification of prime nodes.

We have presented above the principles of a fully dynamic algorithm for modular decomposition of cographs, these can be also done for permutation graphs and interval graphs using
their geometric representation \cite{CP06, Cre09}.
Fully dynamic modular decomposition for the general case is still an open problem.

For some applications one wants to extend the notion of module to some notion of  \textit{approximative} module, for which we want to extend the notion having the same behaviour outside of the module. Several attemps have already been considered  \cite{BHLM09}. The main difficulty is to find an interesting extension of module polynomially tractable, since many of the natural extensions yield to NP-complete problems  \cite{FP03}.

\newcommand{\etalchar}[1]{$^{#1}$}



\end{document}